\documentclass[11pt]{article}

\usepackage{amsthm}
\usepackage{graphicx} 
\usepackage{array} 

\usepackage{amsmath, amssymb, amsfonts, verbatim}
\usepackage{hyphenat,epsfig,subfigure,multirow}
\usepackage{nicefrac}
\usepackage{paralist}
\usepackage{enumitem}

\usepackage[usenames,dvipsnames]{xcolor}
\usepackage[ruled]{algorithm2e}

\DeclareFontFamily{U}{mathx}{\hyphenchar\font45}
\DeclareFontShape{U}{mathx}{m}{n}{
      <5> <6> <7> <8> <9> <10>
      <10.95> <12> <14.4> <17.28> <20.74> <24.88>
      mathx10
      }{}
\DeclareSymbolFont{mathx}{U}{mathx}{m}{n}
\DeclareMathSymbol{\bigtimes}{1}{mathx}{"91}

\usepackage{tcolorbox}
\tcbuselibrary{skins,breakable}
\tcbset{enhanced jigsaw}

\usepackage[normalem]{ulem}
\usepackage[compact]{titlesec}

\definecolor{DarkRed}{rgb}{0.5,0.1,0.1}
\definecolor{DarkBlue}{rgb}{0.1,0.1,0.5}

\usepackage{nameref}
\definecolor{ForestGreen}{rgb}{0.1333,0.5451,0.1333}
\definecolor{Red}{rgb}{0.9,0,0}
\usepackage[linktocpage=true,
	pagebackref=true,colorlinks,
	linkcolor=DarkRed,citecolor=ForestGreen,
	bookmarks,bookmarksopen,bookmarksnumbered]
	{hyperref}
\usepackage[noabbrev,nameinlink]{cleveref}
\crefname{property}{property}{Property}
\creflabelformat{property}{(#1)#2#3}
\crefname{equation}{eq}{Eq}
\creflabelformat{equation}{(#1)#2#3}

\usepackage{bm}
\usepackage{url}
\usepackage{xspace}
\usepackage[mathscr]{euscript}

\usepackage{tikz}
\usetikzlibrary{arrows}
\usetikzlibrary{arrows.meta}
\usetikzlibrary{shapes}
\usetikzlibrary{backgrounds}
\usetikzlibrary{positioning}
\usetikzlibrary{decorations.markings}
\usetikzlibrary{patterns}
\usetikzlibrary{calc}
\usetikzlibrary{fit}
\usetikzlibrary{snakes}

\usepackage{mdframed}

\usepackage[noend]{algpseudocode}
\makeatletter
\def\BState{\State\hskip-\ALG@thistlm}
\makeatother

\usepackage{cite}
\usepackage{enumitem}

\usepackage[margin=1in]{geometry}

\newtheorem{theorem}{Theorem}
\newtheorem{lemma}{Lemma}[section]
\newtheorem{proposition}[lemma]{Proposition}

\newtheorem{claim}[lemma]{Claim}
\newtheorem{fact}[lemma]{Fact}

\newtheorem{definition}[lemma]{Definition}

\newtheorem{problem}{Problem}

\newtheorem*{theorem*}{Theorem}
\newtheorem*{claim*}{Claim}
\newtheorem*{proposition*}{Proposition}
\newtheorem*{lemma*}{Lemma}
\newtheorem*{problem*}{Problem}

\crefname{lemma}{Lemma}{Lemmas}
\crefname{claim}{Claim}{Claims}

\newtheorem{mdresult}{Result}
\newenvironment{result}{\begin{mdframed}[backgroundcolor=lightgray!40,topline=false,rightline=false,leftline=false,bottomline=false,innertopmargin=2pt]\begin{mdresult}}{\end{mdresult}\end{mdframed}}

\newtheorem{remark}[lemma]{Remark}

\newtheoremstyle{restate}{}{}{\itshape}{}{\bfseries}{~(restated).}{.5em}{\thmnote{#3}}
\theoremstyle{restate}

\allowdisplaybreaks

\renewcommand{\qed}{\nobreak \ifvmode \relax \else
      \ifdim\lastskip<1.5em \hskip-\lastskip
      \hskip1.5em plus0em minus0.5em \fi \nobreak
      \vrule height0.75em width0.5em depth0.25em\fi}

\newcommand{\tvd}[2]{\ensuremath{\norm{#1 - #2}_{\mathrm{tvd}}}}

\newcommand{\eps}{\ensuremath{\varepsilon}}

\newcommand{\Bracket}[1]{\Big[#1\Big]}
\newcommand{\bracket}[1]{\left[#1\right]}
\newcommand{\paren}[1]{\ensuremath{\left(#1\right)}\xspace}
\newcommand{\card}[1]{\left\vert{#1}\right\vert}

\newcommand{\norm}[1]{\ensuremath{\|#1\|}}

\newcommand{\expect}[1]{\Exp\bracket{#1}}
\newcommand{\var}[1]{\textnormal{Var}\bracket{#1}}

\newcommand{\set}[1]{\ensuremath{\left\{ #1 \right\}}}

\newcommand{\ALG}{\ensuremath{\mbox{\sc alg}}\xspace}

\newcommand{\TCP}{\ensuremath{\textnormal{\emph{TCP}}}\xspace}

\DeclareMathOperator*{\Exp}{\ensuremath{{\mathbb{E}}}}
\DeclareMathOperator*{\Prob}{\ensuremath{\textnormal{Pr}}}
\renewcommand{\Pr}{\Prob}

\newcommand{\Ex}{\Exp}

\newenvironment{tbox}{\begin{tcolorbox}[
		enlarge top by=5pt,
		enlarge bottom by=5pt,
		 breakable,
		 boxsep=0pt,
                  left=4pt,
                  right=4pt,
                  top=10pt,
                  arc=0pt,
                  boxrule=1pt,toprule=1pt,
                  colback=white
                  ]
	}
{\end{tcolorbox}}

\newcommand{\event}{\ensuremath{\mathcal{E}}}

\newcommand{\kl}[2]{\ensuremath{\mathbb{D}(#1~||~#2)}}
\newcommand{\II}{\ensuremath{\mathbb{I}}}

\newcommand{\mireal}[1][]{
  \ifx\relax#1\relax%
    \II(\mione \,; \mitwo)%
  \else%
    \II(\mione \,; \mitwo\mid #1)%
  \fi
}

\newcommand{\vareps}{\varepsilon}

\newcommand{\dist}{\ensuremath{\mathcal{D}}}

\newcommand{\hindex}[1]{\textsf{h(#1)}}

\renewcommand{\leq}{\leqslant}
\renewcommand{\geq}{\geqslant}

\title{Asymptotically Optimal Bounds for Estimating H-Index in Sublinear Time with Applications to Subgraph Counting}
\author{Sepehr Assadi\footnote{(\texttt{sepehr@assadi.info}) Department of Computer Science, Rutgers University. Research supported in part by a NSF CAREER Grant CCF-2047061, a Google Research gift, and a Fulcrum award from Rutgers Research Council.} \and 
Hoai-An Nguyen\footnote{(\texttt{hnn14@scarletmail.rutgers.edu}) Department of Computer Science, Rutgers University. Research supported in part by a NSF CAREER Grant CCF-2047061.}}

\date{}

\begin{document}
\maketitle

\pagenumbering{roman}

\begin{abstract}
	The \emph{h-index} is a metric used to measure the impact of a user in a publication setting, such as a member of a social network with many highly liked posts or a researcher in an academic domain with many highly cited publications. Specifically, the $h$-index of a user is the largest integer $h$ such that at least $h$ publications of the user have at least $h$ units of positive feedback. 

\medskip

	We design an algorithm that, given query access to the $n$ publications of a user and each publication's corresponding positive feedback number, outputs a $(1\pm \eps)$-approximation of the $h$-index
	of this user with probability at least $1-\delta$ in time
	\[
	O\paren{\frac{n \cdot \ln{(1/\delta)}}{\eps^2 \cdot h}},
	\]
	where $h$ is the actual $h$-index which is unknown to the algorithm a-priori. We then design a novel lower bound technique 
	that allows us to prove that this bound is in fact \textbf{asymptotically optimal} for this problem in \textbf{all parameters} $n,h,\eps,$ and $\delta$. 

\medskip

	Our work is one of the first in sublinear time algorithms that addresses obtaining asymptotically optimal bounds, especially in terms of the error and confidence parameters. As such, we focus on 
	designing novel techniques for this task. In particular, our lower bound technique seems quite general -- to showcase this, we also 
	use our approach to prove an asymptotically optimal lower bound for the problem of estimating the number of triangles in a graph in sublinear time, 
	which now is also optimal in the error and confidence parameters. This result improves upon prior lower bounds of Eden, Levi, Ron, and Seshadhri (FOCS'15) for this problem, 
	as well as multiple follow-ups that extended this lower bound to other subgraph counting problems. 
\end{abstract}


\clearpage

\setcounter{tocdepth}{2}
\tableofcontents

\clearpage

\pagenumbering{arabic}
\setcounter{page}{1}


\newcommand{\ORTP}{\ensuremath{\textnormal{\emph{PTP}}}\xspace}

\newcommand{\R}[2]{\ensuremath{R_{#2}(#1)}}

\newcommand{\D}[3]{\ensuremath{D_{#2,#3}(#1)}}

\renewcommand{\dist}{\ensuremath{D}}

\renewcommand{\ALG}{\mathcal{A}}

\newcommand{\QA}[1]{\ensuremath{Q}_{\ALG}(#1)}

\newcommand{\Bern}[1]{\mathcal{B}(#1)}

\renewcommand{\hindex}[1]{\ensuremath{\textnormal{\textsf{h}}(#1)}\xspace}

\section{Introduction}\label{sec:intro}

The \emph{Hirsch} index, or $h$-index for short, is a metric used to measure the impact of a researcher’s publications~\cite{Hirsch05}. It is an integer that considers both the number of  publications and citations a researcher has and is used in a number of contexts including consideration for grants and job opportunities. We can abstract out this problem by modeling each individual researcher as an array $A[1 : n]$ where $n$ is the number of papers they have published and $A[i]$ is the number of citations paper $i \in [n]$ has. The $h$-index of $A$ is then defined as follows. 

\begin{definition}
The \textbf{\emph{h-index}} of an array $A[1:n]$, denoted by $\hindex{A}$, is the \underline{maximum} integer $h$ such that $A[1 : n]$ has at least $h$ indices, $i_j$, where for each $j\in[h], A[i_j] \geq h$.
\end{definition}

There are simple algorithms that can compute the value of $\hindex{A}$ for any given array $A$ in $O(n)$ time. For instance, we can change each entry of $A[i]$ to $\min\set{A[i],n}$ without changing $\hindex{A}$ (since $\hindex{A} \leq n$), and then
run counting sort on $A$ in linear time to sort $A$ in decreasing order. We can then make another pass over $A$ and output the largest index $i \in [n]$ such that $A[i] \geq i$ which will be equal to $\hindex{A}$ now that $A$ is sorted. 
This solves the $h$-index problem in $\Theta(n)$ time. 

The question we focus on in this paper is whether we can solve this problem even faster than reading the entire input, namely, via a \emph{sublinear time} algorithm, 
assuming we can read each single entry of $A$ in $O(1)$ time. There are easy observations that show that the answer to this question is \emph{No} without relaxing the problem: 
deterministic algorithms cannot solve this problem in sublinear time even approximately, and randomized algorithms cannot find an exact answer\footnote{A deterministic algorithm 
running in $o(n)$ time cannot distinguish between an array $A$ which is all zeros and an array $B$ obtained from $A$ by making $n/2$ entries have value $n/2$ instead. This is because the first $n/2$ queries of the algorithm to indices of $A$ or $B$ can be $0$ 
in both cases. Yet, we have $\hindex{A} = 0$ and $\hindex{B} = n/2$. 

 Similarly, a randomized algorithm running in $o(n)$ time cannot distinguish between an array $A$ filled with value $n$ and an array $B$ obtained from $A$ 
by changing exactly one of the entries to $n-1$ instead. This can be proven, say, by using the $\Omega(n)$ lower bound on the query complexity of OR~\cite{BuhrmanW02}. In this case $\hindex{A}=n$ and $\hindex{B}=n-1$.}. 
Such observations however are commonplace when it comes to sublinear time algorithms. Thus, our goal in this paper is to solve this problem allowing both randomization and approximation. 

\begin{result}\label{res:main}
	There is an algorithm that for any array $A[1:n]$ and any $\eps,\delta \in (0,1)$, with probability at least $1-\delta$, 
	outputs an estimate $\tilde{h}$ such that ${|\tilde{h} - \hindex{A}|} \leq \eps \cdot \hindex{A}$ in $O(\frac{n \cdot \ln{(1/\delta)}}{\eps^2 \cdot \hindex{A}})$ time. 
	Moreover, we prove that this algorithm is {asymptotically optimal} in \underline{all} parameters involved. 
\end{result}

\Cref{res:main} gives a randomized sublinear time algorithm for a $(1\pm \eps)$-approximation of the $h$-index problem, where the runtime improves 
depending on the value of the $h$-index itself. This is quite common in sublinear time algorithms; see, e.g.~\cite{EdenLRS15,EdenRS18,AssadiKK19} for estimating the number of subgraphs,~\cite{BishnuGMP21} for minimum cut, 
or~\cite{EdenR18,FichtenbergerG020,EdenMR21,TetekT22} for sampling small subgraphs, among others. In all the aforementioned examples, such dependences are necessary, which is also the case for ours by the lower bound we prove. 

Our~\Cref{res:main}, however, is quite novel from a different perspective: the obtained bounds are asymptotically optimal in \emph{all} the parameters of the problem, including $\eps$ and $\delta$. We are not aware of any prior work with such strong guarantees
as we will discuss in more detail in the next subsection. Moreover, as a corollary of our techniques in proving the lower bound for~\Cref{res:main} with dependence on both $\eps$ and $\delta$, 
we also obtain an asymptotically optimal lower bound for the well-studied problem of counting triangles in sublinear time that now matches the dependence on $\eps$ and $\delta$ as well, improving upon
the prior work in~\cite{EdenLRS15,EdenR18b,AssadiKK19}.

\subsection{Key Motivations}

There are two key, yet disjoint, motivations behind our work that we elaborate on below.  

\paragraph*{Measuring ``impact'' quickly.} Consider any ``publication setting'' that allows for user feedback. This can range from social networks with users posting topics and others liking them
all the way to the academic domain with researchers publishing papers and others citing them. A question studied frequently in social sciences is how to measure the ``impact'' of a single user in such a setting for many different contexts, including identifying impactful users for marketing or propagating information; see, e.g.~\cite{RiquelmeC16} and the references therein. 

One of the well-accepted measures of impact in these publication settings is the $h$-index measure we study in this paper~\cite{Hirsch05,RiquelmeC16}. Given the ubiquity of massive publication settings and their evolving nature, say, social networks, 
we need algorithms that are able to compute the $h$-index of different users efficiently; see, e.g.~\cite{hIndexStreaming} that design such algorithms in the closely related \emph{streaming} model (which focuses on the space usage of algorithms instead of their time). Thus, a key motivation behind our~\Cref{res:main} is to provide a time-efficient algorithm for this purpose. In general, it seems like a fascinating area of research to 
obtain efficient algorithms for measuring various notions of impact in these massive publication settings in parallel to the line of work, e.g., in~\cite{RiquelmeC16}, that searches for the ``right'' measure itself. 

In addition, the $h$-index measure we study in this paper has numerous applications within network science. In~\cite{edenJain18}, it is shown that when the $h$-index of a network (defined on degrees of vertices) 
is large enough, their algorithm for approximating the degree distribution runs in sublinear time. In~\cite{luZhou16}, the focus is on computing coreness through iteratively using an operator that can calculate the $h$-index of any node to identify influential nodes: an important step in understanding a network's dynamics and structure. Neither work focuses on computing the $h$-index itself efficiently, so the use of our algorithm could help prevent impractical runtimes. Building on~\cite{luZhou16}, ~\cite{sariyuceSP18} generalizes using an iterative $h$-index operator for truss and nucleus decomposition to find dense subgraphs. They use the classical linear algorithm for calculating the $h$-index, which therefore leaves the opportunity to use our algorithm to achieve better efficiency.  

\paragraph*{Asymptotically optimal sublinear time algorithms.} Traditionally, the work on sublinear time algorithms have been rather cavalier with the dependence on the approximation parameter $\eps$, confidence parameter $\delta$, and logarithmic factors. It is certainly important to focus on the ``high order terms'' in the complexity of problems, say, in numerous works on subgraph counting; see, e.g.,~\cite{EdenLRS15,EdenRS18,EdenRS20} and references therein. 
However, as already observed in~\cite{Goldreich17}: ``the dependence of the complexity on the approximation parameter is a key issue''. For instance, in any $(1\pm \eps)$-approximation 
algorithm, for a typical value of $\eps \sim 1\%$, one extra factor of $1/\eps$ in the runtime translates to roughly a $100x$ slower algorithm, which is almost always a deal breaker for the practical purposes of sublinear time algorithms!
Similar considerations also apply, but perhaps to a lower extent, to having a large dependence on logarithmic factors instead of asymptotically optimal bounds. In terms of the confidence parameter, $\delta$, the runtime dependence of sublinear time algorithms almost always includes the term $\ln(1/\delta)$. It is important for practical considerations to determine whether this dependence is necessary. 

Despite this, such considerations have not been studied in sublinear time algorithms. The only prior work we are aware of is the very recent work of~\cite{TetekT22} that improved the $O(\nicefrac{1}{\sqrt{\eps}})$-dependence of the
algorithm of~\cite{EdenR18} for sampling edges $\eps$-point-wise close to uniform to an $O(\log{(1/\eps)})$-dependence. This is in stark contrast with the large body of work 
in  related areas such as streaming~\cite{KaneNW10b,LiW13,BravermanKSV14}, graph streaming~\cite{NelsonY19,AssadiS22}, compressed sensing~\cite{PriceW11,PriceW13}, sampling~\cite{KapralovNPWWY17}, and 
dynamic graph algorithms~\cite{Solomon16,HenzingerP20,BhattacharyaGKL22} which put emphasis on obtaining asymptotically optimal algorithms and lower bounds on all parameters. 

In light of this discussion, another key motivation of our work has been to use the $h$-index problem as a \emph{medium} for designing general techniques 
for obtaining asymptotic bounds for sublinear time algorithms in general. For instance, our algorithm involves careful subroutines that side-step typical ``binary search'' approaches in prior work that results in additional $O(\eps^{-1} \cdot \log{n})$ terms
in the runtimes of algorithms; another example is a more careful analysis of the error that bypasses a trivial union bound which leads to additional $O(\log{n})$ factors. More importantly, we design a new 
technique, based on a new query complexity result that we establish, that allows us to prove lower bounds that depend on both parameters $\eps$ and $\delta$. This approach can now be used to replace 
prior sublinear time lower bounds both based on ad-hoc arguments as in~\cite{EdenLRS15} or communication complexity as in~\cite{EdenR18,AssadiKK19}. As a result, 
we also obtain asymptotically optimal lower bounds for the problem of counting triangles in a graph that now matches the dependence on $\eps$ and $\delta$ as well, improving upon
the prior work in~\cite{EdenLRS15,EdenR18b,AssadiKK19}.



\section{Preliminaries}\label{sec:prelim}

\paragraph{Notation.} For any integer $t \geq 1$, we define $[t] := \set{1,2,\ldots,t}$. For any $p \in (0,1)$, we use $\Bern{p}$ to denote the \emph{Bernoulli} distribution with mean $p$. For a set $S$ of integers,
we write $i \in_R S$ to mean $i$ is chosen uniformly at random from $S$.

\subsection{Basics of Query Complexity}

We use the basics of query complexity to establish our lower bounds on the runtime of sublinear algorithms (as the number of queries made to the input is always a lower bound on the runtime). 

Let $f: \set{0,1}^{n} \mapsto \set{0,1}$ be any Boolean function. A query algorithm for $f$ on any input $x$ can \emph{query} the values of $x_i$ for $i \in [n]$ and determine the value of $f(x)$ with a minimal number of queries. We will work
with the following definitions: 
\begin{itemize}
	\item \textbf{Randomized query complexity}: For any $\delta \in (0,1)$, $\R{f}{\delta}$ denotes the worst-case number of queries made by the best \underline{randomized} algorithm that computes $f$ on any input with probability of success at least $1-\delta$. 
	\item \textbf{Distributional query complexity}: For any $\delta \in (0,1)$ and any distribution $\mu$ on $\set{0,1}^n$, $\D{f}{\mu}{\delta}$ denotes the worst-case number of queries made by the best \underline{deterministic} algorithm that
	computes $f$ on inputs sampled from $\mu$ with probability of success at least $1-\delta$. 
\end{itemize}
Yao's minimax principle~\cite{Yao77} relates these two measures. 
\begin{proposition}[Yao's minimax principle~\cite{Yao77}] \label{prop:yao}
For any $f: \set{0,1}^{n} \mapsto \set{0,1}$ and  $\delta \in (0,1)$:  
\begin{enumerate}[label=$(\roman*)$]
\item \emph{Easy direction (averaging argument):} For any distribution $\mu$ on $\set{0,1}^n$, $\D{f}{\mu}{\delta} \leq \R{f}{\delta}$. 
\item \emph{Hard direction (duality):} There is some distribution $\mu^*$ on $\set{0,1}^n$ such that $\D{f}{\mu^*}{\delta} = \R{f}{\delta}$. 
\end{enumerate}
\end{proposition}

\subsection{Basic Probabilistic Tools}

We use the linearity of variance of independent random variables. 

\begin{fact}\label{fact:variance}
	For any two \textbf{independent} random variables $X$ and $Y$, 
		$\var{X+Y} = \var{X} + \var{Y}$. 
\end{fact}

\noindent
The following proposition lists the standard concentration inequalities we use in this paper. 

\begin{proposition}[Concentration Inequalities; cf.~\cite{DubhashiP09}]\label{prop:conc}
~
\begin{enumerate}[label=$(\roman*)$]
	\item \emph{Chebyshev's inequality:} For any random variable $X$ and $t > 0$, 
	\begin{align*}
		\Pr\paren{\card{X - \expect{X}} \geq t} \leq \frac{\var{X}}{t^2}. 
	\end{align*}
	\item \emph{Chernoff bound:} Suppose $X_1,\ldots,X_n$ are $n$ independent random variables in $[0,1]$ and define $X := \sum_{i=1}^{n} X_i$. Then, for any $\eps \in (0,1)$ and $\mu \geq \expect{X}$, 
	\[
		\hspace{-15pt} \Pr\paren{X > (1+\eps) \cdot \mu} \leq \exp\paren{-\frac{\eps^2 \cdot \mu}{3}} ~ \textnormal{and} ~  \Pr\paren{X < (1-\eps) \cdot \mu} \leq \exp\paren{-\frac{\eps^2 \cdot \mu}{3}}.
	\]
	Moreover, for any $t \geq 1$ and $\mu \geq \expect{X}$, 
	$
	\Pr\paren{\card{X - \expect{X}} \geq t \cdot \mu} \leq 2 \cdot \exp\paren{-\frac{t \cdot \mu}{3}}.
	$ 
\end{enumerate}
\end{proposition}

\subsection{Measures of Distance Between Distributions}\label{prelim-sec:prob-distance}

We use two main measures of distance (or divergence) between distributions, namely  the \emph{total variation distance} and the \emph{Kullback-Leibler divergence} (KL-divergence). 

\paragraph{Total variation distance.} We denote the total variation distance between two distributions $\mu$ and $\nu$ on the same 
support $\Omega$ by $\tvd{\mu}{\nu}$, defined as: 
\begin{align}
\tvd{\mu}{\nu}:= \max_{\Omega' \subseteq \Omega} \paren{\mu(\Omega')-\nu(\Omega')} = \frac{1}{2} \cdot \sum_{x \in \Omega} \card{\mu(x) - \nu(x)}.  \label{prelim-eq:tvd}
\end{align}
\noindent
We use the following basic property of total variation distance. 
\begin{fact}\label{fact:tvd-sample}
	Given a single sample, s, chosen uniformly from one of the distributions $\mu$ or $\nu$, the best probability of successfully deciding whether $s$ came from $\mu$ or $\nu$ 
	is $\frac12 + \frac12\cdot\tvd{\mu}{\nu}$.
\end{fact}

\paragraph{KL-divergence.} For two distributions $\mu$ and $\nu$ over the same probability space, the \emph{Kullback-Leibler divergence} between $\mu$ and $\nu$ is denoted by $\kl{\mu}{\nu}$ and defined as: 
\begin{align}
\kl{\mu}{\nu}:= \Ex_{a \sim \mu}\Bracket{\log\frac{\Pr_\mu(a)}{\Pr_{\nu}(a)}}. \label{prelim-eq:kl}
\end{align}
\noindent
A key property of KL-divergence is that it satisfies a chain rule. 

\begin{fact}[Chain rule for KL-divergence; c.f.~\cite{CoverT06}]\label{fact:chainRuleKL}
Given two distributions $p(x_1,\ldots,x_t)$ and $q(x_1,\ldots,x_t)$ on $t$-tuples, we have, 
\[
	\kl{p}{q} = \sum_{i=1}^{t} \Exp_{p(x_{<i})} \kl{p(x_i \mid x_{<i})}{q(x_i \mid x_{<i})}.
\]
In particular, if $p$ and $q$ are product distributions, then,
$
	\kl{p}{q} = \sum_{i=1}^{t} \kl{p(x_i)}{q(x_i)}.
$
\end{fact} 
The following gives a simple upper bound for the KL-divergence of two Bernoulli distributions. 

\begin{proposition}[KL-divergence on Bernoulli distributions; c.f.~{\cite[Theorem 5]{GibbsS02}}] \label{prop:bernKL} 
For any $0 < p,q < 1$, the following is true: 
\begin{align*}
\kl{\Bern{p}}{\Bern{q}} \leq \dfrac{(p-q)^2}{q \cdot (1-q)}.
\end{align*}
\end{proposition}

We shall also use the following extension of Pinsker's inequality. 

\begin{proposition}[c.f.~{\cite[p. 88-89]{paraEstimation}}]\label{prop:pinsker} 
Given distributions $\mu$ and $\nu$ over a discrete support, 
\[
\tvd{\mu}{\nu} \leq 1 - \dfrac{1}{2} \exp\paren{-\kl{\mu}{\nu}}.
\] 
\end{proposition}


\section{The Algorithm}\label{sec:alg}

We describe our main algorithm for the $h$-index problem in this section. 

\begin{theorem}\label{thm:alg}
There exists a sublinear time algorithm that, given query access to an integer array $A[1:n]$, approximation and confidence parameters $\vareps,\delta \in (0,1)$, with probability at least $1-\delta$ outputs an estimate $\tilde{h}$ of $\hindex A$ such that $|\tilde{h} - \hindex A| \leq \vareps \cdot \hindex A$ in $O(\dfrac{n \cdot \ln(1/\delta)}{\varepsilon^2 \cdot \hindex A})$ time. 
\end{theorem}

The algorithm in~\Cref{thm:alg} is a combination of a ``weak'' and ``strong'' estimator that we design. The weak estimator only outputs whether $\hindex A$ is at least as large as a given threshold, but it is efficient and can be used to provide a lower bound on $\hindex A$. The strong estimator, which has a slower runtime, then uses the lower bound to output an estimate of $\hindex A$. In the next two subsections, we present these two estimators and then conclude the proof of~\Cref{thm:alg} through a careful combination of them that preserves the asymptotic runtime of the overall algorithm. 
\subsection{A Weak Estimator}

We present an algorithm that determines with high probability whether $\hindex A$ is at least as large as a given threshold. 

\begin{lemma} \label{lem:weakEstimator}
There exists a sublinear time algorithm that, given query access to an integer array $A[1:n]$ and an integer $T \geq 1$, in $O(n/T)$ time outputs an answer satisfying the following: 
\begin{enumerate}[label=$(\roman*)$]
	\item if $\hindex{A} \geq T$, the answer is \emph{Large} with probability at least $1-1/16$;
	\item if $\hindex{A} < T/4$, the answer is \emph{Small} with probability at least $1-{\hindex{A}}/{(4T)}$; 
	\item either \emph{Small} or \emph{Large} can be outputted in the remaining cases. 
\end{enumerate}
\end{lemma}

Let us point out the asymmetric guarantee of the algorithm: it does not underestimate $\hindex{A}$ with a certain constant probability while it does not overestimate $\hindex{A}$ with probability proportional 
to the ``rate'' of overestimation. This guarantee will be crucial in our final algorithm. We also note that the guarantee on the runtime of the algorithm is deterministic.

\subsubsection{The Algorithm}

At a high level, our algorithm, \texttt{\hyperref[sec:weak]{h-index-weak-estimator}}, queries random indices from $A$ and calculates the proportion of those indices that are above a threshold representing the mid-point between a $h$-index of $T/4$ and $T$. If the proportion is below the threshold, the algorithm outputs \emph{Small}; otherwise, it outputs \emph{Large}.  The algorithm is formally as follows. 

\begin{tbox}
		\texttt{\hyperref[sec:weak]{h-index-weak-estimator}($A[1:n]$, $T$)}. 
		\begin{enumerate}
            \item Sample $k := {64\cdot n}/{T}$ indices $S$ independently and uniformly with repetition from $[n]$. 
            \item Let $X$ denote the number of indices $i \in S$ such that $A[i] \geq T$. 
            \item If $X \geq {k \cdot T}/{(2n)}$, output \emph{Large} and otherwise \emph{Small}. 
		\end{enumerate}
\end{tbox}

The runtime of \texttt{\hyperref[sec:weak]{h-index-weak-estimator}} is simply $O(n/T)$ as we are sampling these many indices in $S$ and then for each $i \in S$, we need to query $A[i]$; counting the value of $X$ and outputting 
the answer can also be done in $O(n/T)$ time, which bounds the runtime as desired. 
\subsubsection{The Analysis}

We now analyze the correctness of the algorithm. For any $j \in [k]$, define an indicator random variable $X_j$ which is $1$ iff the $j$-th sample in $S$, namely, $i_j \in [n]$, satisfies $A[i_j] \geq T$. This way, 
for the counter $X$ in the algorithm, we have $X = \sum_{j=1}^{k} X_j$. Recall that the output of the algorithm depends on the value of $X$. In the following, we will separately consider
the value of $X$ in the case when the output is supposed to be \emph{Large} versus when it is supposed to be \emph{Small}.

\paragraph*{Case I: the ``Large'' case.} We first consider the case when the output should be \emph{Large}, or when $\hindex{A} \geq T$. In this case, 
\begin{align}
	\expect{X} = \sum_{j=1}^{k} \expect{X_j} = \sum_{j=1}^{k} \Pr_{i_j \in_R [n]}\paren{A[i_j] \geq T} \geq k \cdot \frac{T}{n}, \label{eq:case1-expect}
\end{align}
since $A$ consists of at least $T$ indices with value $\geq T$ when $\hindex{A} \geq T$, and we are sampling indices $i_j \in [n]$ for $j \in [k]$ uniformly at random. 
We can similarly bound the variance of $X$ using~\Cref{fact:variance} since variables $X_{j}$ for $j \in [k]$ are independent, and thus, 
\begin{align}
	\var{X} = \var{\sum_{j=1}^{k} {X_j}} = \sum_{j=1}^{k} \var{X_j} \leq \sum_{j=1}^{k} \expect{X_j^2} =  \sum_{j=1}^{k} \expect{X_j} = \expect{X} \label{eq:case1-var},
\end{align}
where the second to last equality is because for all $j \in [k]$, $X_j$ is an indicator random variable. 

We use Chebyshev's inequality (\Cref{prop:conc}) to finalize the proof of this case. 

\begin{claim} \label{clm:errorLargeE}
When $\hindex A \geq T$, we have $\Pr\paren{\textnormal{algorithm outputs \emph{Small}}} \leq 1/16$.
\end{claim}
\begin{proof}
Recall that the algorithm outputs \emph{Small} iff $X < k \cdot T/(2n)$. By Chebyshev's inequality (\Cref{prop:conc}) with $t= \expect{X}/2$,~\Cref{eq:case1-expect}, and~\Cref{eq:case1-var}, we have, 
\[
	\Pr\paren{X <  \frac{k \cdot T}{2n}} \leq \Pr\paren{\card{X-\expect{X}} > \frac{\expect{X}}{2}} \leq \frac{4 \cdot \var{X}}{\expect{X}^2} \leq \frac{4}{\expect{X}} \leq \frac{1}{16},
\]
where the last inequality is by~\Cref{eq:case1-expect} and the value of $k = 64 \cdot n/T$ in the algorithm. 
\end{proof}
This claim is now enough to establish property $(i)$ in~\Cref{lem:weakEstimator}. 

\paragraph*{Case II: the ``Small'' case.} We now consider the case when the output should be \emph{Small}, namely, when $\hindex{A} < T/4$. In this case, 
we have, 
\begin{align}
	\expect{X} = \sum_{j=1}^{k} \expect{X_j} = \sum_{j=1}^{k} \Pr_{i_j \in_R [n]}\paren{A[i_j] \geq T} < k \cdot \frac{T}{4n}, \label{eq:case2-expect}
\end{align}
as there are less than $T/4$ indices in $A$ with value $\geq T$ when $\hindex{A} < T/4$, and we are sampling indices $i_j \in [n]$ for $j \in [k]$ uniformly at random. 
We will also bound the variance of $X$ similarly to~\Cref{eq:case1-var} but in a slightly more careful manner. By~\Cref{fact:variance}, since variables $X_{j}$ for $j \in [k]$ are independent, we have, 
\begin{align}
	\var{X} = \sum_{j=1}^{k} \var{X_j} \leq \sum_{j=1}^{k} \expect{X_j} = \sum_{j=1}^{k} \Pr_{i_j \in_R [n]}\paren{A[i_j] \geq T} \leq k \cdot \frac{\hindex{A}}{n} \label{eq:case2-var},
\end{align}
where in the last inequality, we use the fact that the number of indices in $A$ with value larger than $T$ is at most $\hindex{A}$ (since we already know that $\hindex{A} < T$). 

To conclude the proof, we again use Chebyshev's inequality but with a slightly different analysis. 

\begin{claim}  \label{clm:errorSmallE}
When $\hindex A < T/4$, we have $\Pr\paren{\textnormal{algorithm outputs \emph{Large}}} \leq \hindex{A}/(4T)$.
\end{claim}
\begin{proof}
	The algorithm will output \emph{Large} iff $X \geq k\cdot T/(2n)$. By Chebyshev's inequality (\Cref{prop:conc}) with $t= k \cdot T/(4n)$, and~\Cref{eq:case2-expect},~\Cref{eq:case2-var}, we have, 
\[
	\Pr\paren{X \geq  \frac{k \cdot T}{2n}} \leq \Pr\paren{\card{X - \expect{X}} > \frac{k \cdot T}{4n}} \leq \frac{\var{X}}{(k \cdot T/4n)^2} \leq \frac{(k \cdot \hindex{A}/n)}{(k \cdot T/4n)^2} = \frac{\hindex{A}}{4T}, 
\]
where the last equality is by the choice of $k = 64 \cdot n/T$ in the algorithm. 
\end{proof}

\Cref{lem:weakEstimator} now follows from the previous two claims.

\subsection{A Strong Estimator}
We now present our second intermediate algorithm which outputs an estimate of $\hindex A$ when given the guarantee that $\hindex A$ is at least as large as a given threshold. 
\begin{lemma} \label{lem:algStrong}
There exists a sublinear time algorithm that, given query access to an integer array $A[1:n]$, an integer $T \leq \hindex{A}$, and approximation parameter $\varepsilon \in (0,1)$, 
in  $O(n/(\varepsilon^2 T))$ time outputs an estimate $\tilde{h}$ of $\hindex A$ such that 
$\Pr(|\tilde{h} - \hindex A| \leq \varepsilon \cdot \hindex A) \geq 2/3$.

The guarantee on the runtime of the algorithm holds deterministically even when $T > \hindex{A}$. 
\end{lemma}

We emphasize that while the guarantee on the runtime of the algorithm in~\Cref{lem:algStrong} holds even when $T > \hindex{A}$, we clearly have no guarantee on the 
correctness in this case. 

\subsubsection{The Algorithm}
The algorithm, $\texttt{\hyperref[sec:strong]{h-index-strong-estimator}}$, queries a set of random indices from $A$ and finds a scaled estimate of the $ h$-index. Formally, 

\begin{tbox}
		\texttt{\hyperref[sec:strong]{h-index-strong-estimator}($A[1:n]$, $T$, $\varepsilon$)}.
		\begin{enumerate}
			\item Sample $k := {6n}/{(\varepsilon^2 T)}$ indices $S$ independently and uniformly with repetition from $[n]$. 
			\item Let $B[1:k]$ be an array consisting of integers $A[i]$ for $i \in S$. 
			\item Return\footnote{We will implement this step efficiently in~\Cref{lem:strongRuntime}.} the largest integer $q \in [n]$ such that $q \cdot k/n$ indices in $B$ are at least $q$.  
		\end{enumerate}
\end{tbox}

The first two lines of $\texttt{\hyperref[sec:strong]{h-index-strong-estimator}}$ can be implemented in $O(k) = O(n/(\eps^2 T))$ time in a straightforward way. We show that the last step can also be implemented in
$O(k)$ time. 

\begin{lemma}\label{lem:strongRuntime}
\textnormal{\texttt{\hyperref[sec:strong]{h-index-strong-estimator}}} runs in $O(n/(\varepsilon^2 T))$ time. 
\end{lemma}
\begin{proof}
As argued earlier, we only need to focus on implementing the last step of the algorithm in $O(k)$ time. The problem we want to solve can be seen as a ``discounted'' version of the original $h$-index problem: for $\alpha = k/n$, our goal is to find the largest 
number $q$ in $B$ such that there are at least $\alpha \cdot q$ numbers as large as $q$ in $B$. In the case when $\alpha=1$, this is exactly the $h$-index problem itself on $B$, but for smaller values of $\alpha$, this can change. We show how 
to solve this discounted $h$-index problem in linear time. For the purpose of this problem, we assume all the values in  $B$ are distinct which can be achieved by  using a consistent tie-breaking rule. 

We design an algorithm \texttt{\hyperref[alg:discount]{discounted-h-index}}$(C[1:t],d,\alpha)$ that given any array $C[1:t]$, an integer $d \geq 0$, and parameter $\alpha \in [0,1]$, outputs the largest number $p$ from $C$ such that there are at least $\alpha \cdot p - d$ numbers
with value $\geq p$ in $C$ under the promise that at least one such number always exists in the array. We can then run \texttt{\hyperref[alg:discount]{discounted-h-index}}$(B[1:k],0,k/n)$ to obtain the desired integer $q$ in the third line of $\texttt{\hyperref[sec:strong]{h-index-strong-estimator}}$. 

\begin{tbox}
	$\texttt{discounted-h-index}(C[1:t],d,\alpha)$. 
	\begin{enumerate}[label=$(\roman*)$]
		\item If $t=1$ return $C[1]$; otherwise, continue with the following lines. 
		\item Run the Median-of-Medians algorithm of \cite{medianOfMedians} to find the median $p$ of $C$ and partition $C$ using $p$ so that $p$ appears in position $t/2$. All smaller elements of $C$
		appear in $C[1:t/2)$ and all larger elements appear in $C(t/2:t]$. 
		\item If $t/2 \geq \alpha \cdot p - d$, return $\texttt{discounted-h-index}(C[t/2:t],d+t/2,\alpha)$. Otherwise, return $\texttt{discounted-h-index}(C[1:t/2),d,\alpha)$.
	\end{enumerate}
\end{tbox}

The runtime of this algorithm is $O(t)$ because we are spending $O(t)$ time to run the Median-of-Medians algorithm and another $O(t)$ time to partition $C$ and recurse on an array of size $t/2$. 

We now prove the correctness of this algorithm. The base case of the algorithm for $t=1$ is correct by the promise on the existence of the number in the problem statement. 
Now consider a larger value of $t$ and the median $p$ of $C$. Since $p$ is the median, we know that
there are exactly $t/2$ elements with value $\geq p$ in $C$. 
\begin{itemize}
\item If $t/2 \geq \alpha \cdot p - d$, we get that the ``right'' answer can only be $\geq p$ by definition and thus we only need to search for it in $C[t/2:t]$ -- moreover, 
for any number here, we know it is larger than $t/2$ elements in the original array $C$, and thus, we only need our number in $C[t/2:t]$ to be larger than $\alpha \cdot p - d - t/2$ numbers to satisfy the original requirement over $C$ itself. 
This is returned correctly by induction on the smaller subproblem $\texttt{\hyperref[alg:discount]{discounted-h-index}}(C[t/2:t],d+t/2,\alpha)$ as we maintain the promise of the existence of the ``right'' answer in this subarray. 
\item On the other hand, if $t/2 < \alpha \cdot p - d$, we know that the ``right'' answer is $<p$ by definition and thus we should only search for it in $C[1:t/2)$. In this case, by induction, 
$\texttt{\hyperref[alg:discount]{discounted-h-index}}(C[1:t/2),d,\alpha)$ returns the correct answer as we again maintain the promise of the existence of the ``right'' answer in this subarray. 
\end{itemize}
This concludes the proof of the correctness of $\texttt{\hyperref[alg:discount]{discounted-h-index}}$. Thus, the last step of our main algorithm also runs in $O(k) = O(n/(\eps^2 T))$ time, concluding the 
proof of~\Cref{lem:strongRuntime}. 
\end{proof}

\subsubsection{The Analysis}

We now prove the correctness of $\texttt{\hyperref[sec:strong]{h-index-strong-estimator}}$. We consider each case in which the algorithm may overestimate or underestimate $\hindex{A}$ separately. 

\paragraph*{Probability of overestimation.} We first bound the probability that $\tilde{h} > (1+\eps) \cdot \hindex{A}$. For this event to happen, we need $B$
to have more than $(k/n) \cdot (1+\eps) \cdot \hindex{A}$ indices with a value greater than $(1+\eps) \cdot \hindex{A}$. We bound the probability of this happening in the following. 

For any $j \in [k]$, define an indicator random variable $X_j$ which is $1$ iff the $j$-th sample $i_j \in S$ satisfies $A[i_j] > (1+\eps) \cdot \hindex{A}$. Define $X := \sum_{j=1}^{k} X_j$. By the above discussion, 
\begin{align}
	\Pr\paren{\tilde{h} > (1+\eps) \cdot \hindex{A}} = \Pr(X > (k/n) \cdot (1+\eps) \cdot \hindex{A}). \label{eq:strong1} 
\end{align}
We bound the probability of the RHS of this equation. 

\begin{claim} \label{clm:overestimation}
	$ \Pr\paren{X > ({k}/{n}) \cdot (1+\eps) \cdot \hindex{A}} < 1/6$. 
\end{claim}
\begin{proof}
	Consider the expected value of $X$: 
	\[
		\expect{X} = \sum_{j=1}^{k} \expect{X_j} = \sum_{j=1}^{k} \Pr_{i_j \in_R [n]}\paren{A[i_j] > (1+\eps) \cdot \hindex{A}} \leq k \cdot \frac{\hindex{A}}{n},
	\]
	where the last inequality is because the number of indices with value greater than $\hindex{A}$ in $A$ is at most $\hindex{A}$. For the variance of $X$, using the independence of variables $\set{X_j}_{j \in [k]}$ 
	and~\Cref{fact:variance}, 
	\[
		\var{X} = \var{\sum_{j=1}^{k} X_j} = \sum_{j=1}^{k} \var{X_j} \leq \sum_{j=1}^{k} \expect{X_j} = \expect{X}, 
	\]
	where the inequality holds since $X_j$ is an indicator random variable, so, $\var{X_j} \leq \Exp[{X_j^2}] = \expect{X_j}$. 
	
	By Chebyshev's inequality (\Cref{prop:conc}) with parameter $t=\eps \cdot k \cdot {\hindex{A}}/{n}$, we have, 
	\[
		 \Pr\paren{X > \frac{k}{n} \cdot (1+\eps) \cdot \hindex{A}} \leq \Pr\paren{\card{X-\expect{X}} > \eps k \cdot \frac{\hindex{A}}{n}} < \frac{\var{X}}{(\eps \cdot k \cdot {\hindex{A}}/{n})^2} \leq \frac{n}{\eps^2 \cdot k \cdot {\hindex{A}}} \leq \frac{1}{6},
	\]
	by the choice of $k = 6n/(\eps^2 T) \geq 6n/(\eps^2\, \hindex{A})$ since $T \leq \hindex{A}$. 
\end{proof}

\paragraph*{Probability of underestimation.} We now bound the probability that $\tilde{h} < (1-\eps) \cdot \hindex{A}$. This case is essentially symmetric to the other one and is provided for completeness. 
For this event to happen, we need $B$
to have less than $(k/n) \cdot (1-\eps) \cdot \hindex{A}$ indices with a value of at least $(1-\eps) \cdot \hindex{A}$. We bound the probability of this happening in the following. 

For any $j \in [k]$, define an indicator random variable $Y_j$ which is $1$ iff the $j$-th sample $i_j \in S$ satisfies $A[i_j] \geq (1-\eps) \cdot \hindex{A}$. Define $Y = \sum_{j=1}^{k} Y_j$. By the above discussion, 
\begin{align}
	\Pr\paren{\tilde{h} < (1-\eps) \cdot \hindex{A}} = \Pr\paren{Y < (k/n) \cdot (1-\eps) \cdot \hindex{A}}. \label{eq:strong2} 
\end{align}
We bound the probability of the RHS of this equation. 

\begin{claim} \label{clm:underestimation}
	$ \Pr\paren{Y < ({k}/{n}) \cdot (1-\eps) \cdot \hindex{A}} < 1/6$. 
\end{claim}
\begin{proof}
	Consider the expectation of $Y$: 
	\[
		\expect{Y} = \sum_{j=1}^{k} \expect{Y_j} = \sum_{j=1}^{k} \Pr_{i_j \in_R [n]}\paren{A[i_j] \geq (1-\eps) \cdot \hindex{A}} \geq k \cdot \frac{\hindex{A}}{n},
	\]
	where the last inequality is because the number of indices that is at least some value $v \leq \hindex{A}$ in $A$ is at least $\hindex{A}$. For the variance of $Y$, using the independence of variables $\set{Y_j}_{j \in [k]}$ 
	and~\Cref{fact:variance},
	\[
		\var{Y} = \var{\sum_{j=1}^{k} Y_j} = \sum_{j=1}^{k} \var{Y_j} \leq \sum_{j=1}^{k} \expect{Y_j} = \expect{Y}, 
	\]
	where the inequality holds since $Y_j$ is an indicator random variable, so, $\var{Y_j} \leq \Exp[{Y_j^2}] = \expect{Y_j}$. 
	
	By Chebyshev's inequality (\Cref{prop:conc}) with parameter $t=\eps \cdot \expect{Y}$, we have, 
	\[
		 \Pr\paren{Y < \frac{k}{n} \cdot (1-\eps) \cdot \hindex{A}} \leq \Pr\paren{\card{Y-\expect{Y}} > \eps \cdot \expect{Y}} < \frac{\var{Y}}{(\eps \cdot \expect{Y})^2} \leq \frac{n}{\eps^2 \cdot k \cdot {\hindex{A}}} \leq \frac{1}{6},
	\]
	by the choice of $k = 6n/(\eps^2 T) \geq 6n/(\eps^2\, \hindex{A})$ since $T \leq \hindex{A}$. 
\end{proof}

Combining~\Cref{clm:overestimation} and~\Cref{clm:underestimation} concludes the proof of~\Cref{lem:algStrong}.

\subsection{The Sublinear Time h-Index-Estimator Algorithm}

We now combine our weak and strong estimators to obtain a sublinear time algorithm for estimating the $h$-index and prove~\Cref{thm:alg}. 
The algorithm runs $\texttt{\hyperref[sec:weak]{h-index-weak-estimator}}$ on smaller and smaller thresholds to determine a threshold that tightly lower bounds $\hindex A$. Then, $\texttt{\hyperref[sec:strong]{h-index-strong-estimator}}$ uses that threshold to output an estimate of $\hindex A$. Finally, to ensure a probability of success of at least $1-\delta$, we combine the median/majority trick in a rather non-black-box way using the asymmetric guarantee of 
$\texttt{\hyperref[sec:weak]{h-index-weak-estimator}}$ in part $(ii)$ of~\Cref{lem:weakEstimator}. 

\begin{tbox}
		\texttt{\hyperref[alg:h]{h-index-estimator}($A[1:n]$, $\varepsilon$, $\delta$)}.
\begin{enumerate}
\item Let $r_1 := 7\ln(8/\delta)$ and $r_2 :=  108\ln(8/\delta)$ and initialize $T$ to $n$. 
\item While the \emph{majority} answer of running $\texttt{\hyperref[sec:weak]{h-index-weak-estimator}}$($A$, $T$) $r_1$ times returns \emph{Small}, update $T \leftarrow T/4$. 
\item For the current value of $T$, run $\texttt{\hyperref[sec:strong]{h-index-strong-estimator}}$($A$, $T/16$, $\varepsilon$) $r_2$ times and return the median answer as the final estimate $\tilde{h}$. 
\end{enumerate}
\end{tbox}

We bound the runtime of the algorithm in the following lemma. 

\begin{lemma}
\textnormal{$\texttt{\hyperref[alg:h]{h-index-estimator}}$} runs in $O\Big{(} \dfrac{n \cdot \ln(1/\delta)}{\varepsilon^2 \cdot \hindex A} \Big{)}$ time with probability $1-\delta/2$. 
\end{lemma}

\begin{proof}
The runtime depends on both running \texttt{\hyperref[sec:weak]{h-index-weak-estimator}} on (potentially) multiple thresholds and running \texttt{\hyperref[sec:strong]{h-index-strong-estimator}}. 

We define {$T^*$} \unboldmath as the ``optimal'' threshold: the \emph{first} threshold given to $\texttt{\hyperref[sec:weak]{h-index-weak-estimator}}$ that is not larger than $\hindex A$, namely, $T^* \leq \hindex{A} < 4 \cdot T^*$.  
The following claim bounds the probability that the while-loop in step two of \texttt{\hyperref[alg:h]{h-index-estimator}} does not stop even after iteration $T^*$.
\begin{claim}\label{clm:T*}
	  $\Pr\paren{\textnormal{{$\texttt{\hyperref[alg:h]{h-index-estimator}}$} continues its while-loop beyond $T^*$}} \leq \delta/2$. 
\end{claim}
\begin{proof}
	Assuming \textnormal{$\texttt{\hyperref[alg:h]{h-index-estimator}}$} reaches the iteration $T^*$, we have that $\hindex{A} \geq T^*$. Thus, by part $(i)$ of~\Cref{lem:weakEstimator}, 
	each of the $r_1$ runs of {$\texttt{\hyperref[sec:weak]{h-index-weak-estimator}}$} in this iteration outputs \emph{Small} with probability at most $1/16$. Let $X$ be the random variable 
	that denotes the number of these $r_1$ runs that output \emph{Small}. We thus have $\expect{X} \leq r_1/16$. Moreover, $X$ is a sum of $r_1$ independent random variables
	and thus by the Chernoff bound (\Cref{prop:conc}), we have, 
	\[
		\Pr\paren{X \geq r_1/2} \leq \Pr\paren{\card{X-\expect{X}} \geq 7 \cdot r_1/16} \leq 2 \cdot \exp\paren{-\frac{7 \cdot r_1}{48}} < \delta/2,
	\]
	as $r_1 = 7\ln(8/\delta)$. This concludes the proof as for $\texttt{\hyperref[alg:h]{h-index-estimator}}$ to continue beyond the iteration $T^*$ we need $X \geq r_1/2$, which happens with probability $< \delta/2$. 
\end{proof}

In the following, we condition on the complement of the event in~\Cref{clm:T*} which happens with probability at least $1-\delta/2$ -- this means we have only run the while-loop until at most iteration $T^*$. Let $T_0 = n, T_1 = n/4, \ldots,T_t = n/4^t = T^*$ 
denote the thresholds in these  iterations. By~\Cref{lem:weakEstimator} on the runtime of $\texttt{\hyperref[sec:weak]{h-index-weak-estimator}}$ we have, 
\begin{align*}
	\text{runtime of while-loop} = \sum_{j=0}^{t} O(\frac{n}{T_j}) \cdot O(\ln{(1/\delta)}) &= O\paren{\frac{n}{T^*} \cdot \ln{(1/\delta})} \cdot \sum_{j=0}^{t} \frac{1}{4^j} \\&= O\paren{\frac{n}{\hindex{A}} \cdot \ln{(1/\delta)}},
\end{align*}
since $T^*$ is a $4$-approximation to $\hindex{A}$ by definition and the given geometric series converges. 

Moreover, by~\Cref{lem:algStrong} on the runtime of $\texttt{\hyperref[sec:strong]{h-index-strong-estimator}}$, in this case, we have that the last line of the algorithm
takes $O(\frac{n \cdot \ln{(1/\delta)}}{\eps^2 \cdot T^*}) = O(\frac{n \cdot \ln{(1/\delta)}}{\eps^2 \cdot \hindex{A}})$
time as well, again since $T^*$ is a $4$-approximation to $\hindex{A}$ (computing the medians can be done with the Median-of-Medians algorithm in $O(r_2) = O(\log{(1/\delta)})$ time which is negligible in the above bounds). 

All in all, we have that with probability $1-\delta/2$, the algorithm runs in $O(\frac{n \cdot \ln{(1/\delta)}}{\eps^2 \cdot \hindex{A}})$ time. 
\end{proof}

\subsubsection{The Analysis}

We prove the correctness of our algorithm in this subsection. Consider the parameter $T^*$ defined earlier as the ``optimal'' threshold in the while-loop, meaning that $T^* \leq \hindex{A} < 4 \cdot T^*$.  
There are two potential sources for error: 
\begin{enumerate}
	\item \textbf{Event $\event_{weak}$}: In the while-loop, $\texttt{\hyperref[sec:weak]{h-index-weak-estimator}}$ outputs \emph{Large} for an iteration $T > 16T^*$; assuming this happens, the threshold passed to 
	$\texttt{\hyperref[sec:strong]{h-index-strong-estimator}}$ is not necessarily valid, meaning that it may not be a lower bound on $\hindex{A}$. 
	\item \textbf{Event $\event_{strong}$}: The threshold $T$ obtained by the runs of $\texttt{\hyperref[sec:weak]{h-index-weak-estimator}}$ in the while-loop satisfies $T \leq 16T^*$ and thus is valid, but $\texttt{\hyperref[sec:strong]{h-index-strong-estimator}}$ nevertheless fails to output an accurate estimate 
	of $\hindex{A}$. 
\end{enumerate}

Among these, the probability of the second event is quite easy to bound using~\Cref{lem:algStrong}. Thus, in the following, we focus primarily on proving the first part. 
\begin{claim}\label{lem:while-loop}
	In \textnormal{$\texttt{\hyperref[alg:h]{h-index-estimator}}$}, for any $T = 4^{\ell} \cdot T^*$ for an integer $\ell \geq 2$, 
	\[\Pr\paren{\textnormal{the while-loop terminates at iteration $T$}} \leq \paren{\delta/8}^{\ell-1}.\] 
\end{claim}
\begin{proof}
	Recall that at iteration $T$, we have $r_1$ runs of $\texttt{\hyperref[sec:weak]{h-index-weak-estimator}}$ with threshold $T$. For $j \in [r_1]$, let $E_j$ denote the event that the $j$-th run of 
	$\texttt{\hyperref[sec:weak]{h-index-weak-estimator}}$ outputs \emph{Large}. By the majority rule in the algorithm, we have, 
	\begin{align*}
		\Pr\paren{\textnormal{while-loop terminates at $T$}} &\leq \Pr\paren{\textnormal{$\exists S \subseteq [r_1], \card{S}=\frac{r_1}{2}$ such that $\bigwedge_{j \in S} E_j$}} \\
		&\leq {{r_1}\choose{\nicefrac{r_1}{2}}} \Pr\paren{\textnormal{$\texttt{\hyperref[sec:weak]{h-index-weak-estimator}}$ outputs \emph{Large}}}^{{r_1}/{2}},
	\end{align*}
	where we use the union bound and the independence of the $r_1$ runs of $\texttt{\hyperref[sec:weak]{h-index-weak-estimator}}$. 
	Now, note that since $T \geq 16T^* > 4\hindex{A}$, we can apply Property $(ii)$ of~\Cref{lem:weakEstimator} and have that 
	\[
		 \Pr\paren{\textnormal{$\texttt{\hyperref[sec:weak]{h-index-weak-estimator}}$ outputs \emph{Large}}} \leq \frac{\hindex{A}}{4\,T} \leq \frac{1}{4^{\ell}},
	\]
	by the value of $T$. Combining this with the previous equation gives us 
	\[
		\Pr\paren{\textnormal{while-loop terminates at  $T$}} \leq 2^{r_1} \cdot \paren{\frac{1}{4^{\ell}}}^{r_1/2} = \paren{\frac{1}{2}}^{r_1 \cdot (\ell-1)} < \paren{\frac{\delta}{8}}^{\ell-1}
	\]
	by the choice of $r_1 = 7\ln{(8/\delta})$, concluding the proof. 
\end{proof}

We can now bound the error probability due to event $\event_{weak}$. We have, 
\begin{align*}
	\Pr\paren{\event_{weak}} &\leq \sum_{\ell \geq 2} \Pr\paren{\textnormal{the while-loop terminates at  $T = 4^{\ell} \cdot T^*$}} \\
	&\leq \sum_{\ell \geq 2} \paren{\frac{\delta}{8}}^{\ell-1} \tag{by~\Cref{lem:while-loop}} \\
	&= \frac{(\delta/8)}{1-(\delta/8)}  \tag{as $\sum_{j=1}^{\infty} x^j = \frac{x}{1-x}$ for $x \in (0,1)$} 
	< \frac{\delta}{4}. 
\end{align*}

We now bound the other source of error. Assuming $\event_{weak}$ does not happen,  for the parameter $T$ that the while-loop terminates on, we have $T \leq 16T^* \leq 16\hindex{A}$ by the definition of $T^*$. 
This implies that the parameter $T/16$ passed to $\texttt{\hyperref[sec:strong]{h-index-strong-estimator}}$ is a lower bound on $\hindex{A}$. Thus, by~\Cref{lem:algStrong}, each of the $r_2$ runs of 
$\texttt{\hyperref[sec:strong]{h-index-strong-estimator}}$ outputs a $(1\pm \eps)$-approximation to $\hindex{A}$ with probability at least $2/3$. 
\begin{claim}\label{clm:median-trick}
	$\Pr\paren{\event_{strong} \mid \overline{\event_{weak}}} \leq {\delta}/{4}$. 
\end{claim}
\begin{proof}
For any $j \in [r_2]$, let $X_j$ be an indicator random variable which is $1$ iff the $j$-th run of $\texttt{\hyperref[sec:strong]{h-index-strong-estimator}}$ outputs an estimate which is \emph{not} a $(1\pm \eps)$-approximation to $\hindex{A}$. 
Let $X = \sum_{j=1}^{r_2} X_j$. By~\Cref{lem:algStrong}, we have $\expect{X} \leq r_2/3$. By the Chernoff bound (\Cref{prop:conc}), 
\[
	\Pr\paren{\event_{strong} \mid \overline{\event_{weak}}} \leq \Pr\paren{X \geq r_2/2} \leq \Pr\paren{\card{X - \expect{X}} \geq r_2/6} \leq 2 \cdot \exp\paren{-\frac{r_2}{108}} = \frac{\delta}{4},
\]
by the choice of $r_2 = 108\,\ln{(8/\delta)}$. 
\end{proof}

Therefore, by the union bound, the total probability of error is at most $\delta/4 + \delta/4 = \delta/2$. This concludes the analysis of $\texttt{\hyperref[alg:h]{h-index-estimator}}$.


\section{The Lower Bound}\label{sec:lower} 

We now prove the asymptotic optimality of the bounds obtained by our algorithm in~\Cref{thm:alg}. 

\begin{theorem}\label{thm:lower}
Any algorithm that, given query access to an array $A[1:n]$, approximation parameter $\varepsilon \in (0,1/4)$, and confidence parameter $\delta \in (0,1/100)$, with probability $1-\delta$ uses at most $q$ queries and 
outputs an estimate $\tilde{h}$ such that $|\tilde{h} - \hindex{A} | \leq \eps \cdot \hindex A$ needs to satisfy 
$q = \Omega(\min(n, \frac{n \cdot \ln (1/\delta)}{\varepsilon^2 \cdot \hindex A}))$.
\end{theorem}

To prove~\Cref{thm:lower}, we define a new problem which we call the \emph{Popcount Thresholding Problem (PTP)} and prove a lower bound on its randomized query complexity. We will then perform 
a  reduction from this problem to establish our theorem. 
\begin{remark}
In the following, we focus on proving that when $\hindex A \geq \ln(1/\delta) \cdot 12/\eps^2$, the randomized query complexity is $\Omega((n \cdot \ln(1/\delta))/(\eps^2 \cdot \hindex A))$. We claim this is without loss of generality
when proving~\Cref{thm:lower} because of the following. Let us suppose that $1 \leq \hindex A < \ln(1/\delta) \cdot 12/\eps^2$. There exists some $\eps' > \eps$ and $\delta' > \delta$ such that $\hindex A = \ln(1/\delta') \cdot 12/\eps'^2$. 
For this choice of parameter, we can apply the first part of the result to get a lower bound of $(n\cdot \ln(1/\delta')) / (\eps'^2 \cdot \hindex A) = \Omega(n)$, which implies~\Cref{thm:lower} in this case as well. 
\end{remark}
In passing, we note that \ORTP seems quite a natural and general problem of its own independent interest; we will also use this problem in the subsequent section to 
prove asymptotically optimal lower bounds for the well-studied problem of estimating the number of triangles in a graph.

\subsection{Popcount Thresholding Problem (PTP)} \label{sec:orThresholding}

We define the Popcount Thresholding Problem as follows. 

\begin{problem}\label{prob:ortp}
	In $\ORTP_{m,k,\gamma}$, for integers $m,k,\geq 1$ and parameter $\gamma \in (0,1)$, we are given a string $x \in \set{0,1}^m$ sampled with equal probability from either distribution $D_0$ or $D_1$ defined as follows: 
	\begin{itemize}
	\item In $D_0$, for each index $i \in [m]$, $x_i$ is independently set to $1$ with probability $p_0 := (1-2\gamma) \cdot k/m$;
	\item In $D_1$, for each index $i \in [m]$, $x_i$ is independently set to $1$ with probability  $p_1 := (1+2\gamma) \cdot k/m$. 
	\end{itemize}
	The answer is \emph{Yes} if $x$ was drawn from $D_1$, and it is \emph{No} if $x$ was drawn from $D_0$. 
\end{problem}

We prove the following lemma on the query complexity of \ORTP. 

\begin{lemma}\label{lem:ORTP}
	For any $\gamma \in (0,1/4)$, $\delta \in (0,1/100)$, and integers $m \geq 1$, $\ln{(1/\delta)} \cdot 12/\gamma^2 \leq k \leq m/6$, 
	\[
	\R{\ORTP_{m,k,\gamma}}{\delta} \geq \frac{m \cdot \ln{(1/(4\delta))}}{24\,\gamma^2 \cdot k},
	\]
	 where $R_{\delta}(\cdot)$ denotes the randomized query complexity with error probability $\delta$. 
\end{lemma}
\begin{proof}
To prove~\Cref{lem:ORTP}, we use the easy direction of Yao's minimax principle (\Cref{prop:yao}) which allows us to focus on \emph{deterministic} algorithms for $\ORTP$ its input distribution. As per~\Cref{prob:ortp}, the input distribution is $D = (1/2) \cdot D_0 + (1/2) \cdot D_1$.

For the rest of the proof, let $\ALG$ be any deterministic query algorithm on $\dist$ with the worst-case number of queries 
\begin{align}
q(\ALG) := q < \frac{m \cdot \ln{(1/(4\delta))}}{24\,\gamma^2 \cdot k} \label{eq:choice-q}.
\end{align}
Without loss of generality, we assume that $\ALG$ always makes $q$ queries on any input (by potentially making 
``dummy" queries to reach $q$ if needed). For an input $x \sim \dist$, we use $\QA{x} \in \set{0,1}^q$ to denote the string of answers returned to the query algorithm based on $x$. 

\paragraph*{Distribution of $\QA{x}$.} A key observation is that given only $\QA{x} = (b_1,\ldots,b_q)$, since $\ALG$ 
is a deterministic algorithm, we will learn the value of exactly $q$ specific entries in $x$: $b_1$ is the value of the index of $x$ queried first by $\ALG$, then, $b_2$ is the value of the second index queried by $\ALG$ where the query is uniquely determined after seeing the answer $b_1$ to the first query, and so on and so forth. Thus, for any choice of $\theta \in \set{0,1}$, \emph{conditioned on} $x$ being sampled from $D_{\theta}$, for any $i \in [m]$, \emph{independent} of the value of $(b_1,\ldots,b_{i-1})$, the 
value of $b_i$ is sampled from a Bernoulli distribution with mean $p_{\theta}$. This means that:
\begin{align*}
	&\text{distribution $(\QA{x} \mid D_0)$ is $\Bern{p_0}^{q}$} \qquad \text{and} \qquad \text{distribution $(\QA{x} \mid D_1)$ is $\Bern{p_1}^{q}$}. \label{eq:means}
\end{align*}
The following claim bounds the KL-divergence (\Cref{prelim-eq:kl}) between these two distributions. 

\begin{claim}\label{clm:kl}
	For any $q \geq 1$ and $0 < p_0,p_1 < 1/3$, we have, 
		$\kl{\Bern{p_0}^{q}}{\Bern{p_1}^q} < \ln{(1/(4\delta))}$.
\end{claim}
\begin{proof}
	By the chain rule of KL-divergence and using the fact that both arguments are product distributions (\Cref{fact:chainRuleKL}), we have
	\[
		\kl{\Bern{p_0}^{q}}{\Bern{p_1}^q} = q \cdot \kl{\Bern{p_0}}{\Bern{p_1}}. 
	\]
	Moreover, for each term, using~\Cref{prop:bernKL}, we have
	\[
		\kl{\Bern{p_0}}{\Bern{p_1}} \leq \dfrac{(p_0-p_1)^2}{p_1 \cdot (1-p_1)}  \leq \frac{(4\gamma \cdot k/m)^2}{(1+2\gamma)\cdot k/m \cdot 2/3} \leq 24\gamma^2 \cdot \frac{k}{m} < \frac{\ln{(1/(4\delta))}}{q}, 
	\]
	by the choice of $q$ in~\Cref{eq:choice-q}, concluding the proof. \qed$_{\,\textnormal{\Cref{clm:kl}}}$
	
\end{proof}

Let us now use~\Cref{clm:kl} to conclude the proof of~\Cref{lem:ORTP}. As argued earlier, all the information that is revealed to the algorithm $\ALG$ is the string $\QA{x}$ on an input $x \sim \dist$, and its task is to distinguish whether $x$ is sampled from $\dist_0$ or $\dist_1$. 
By~\Cref{fact:tvd-sample}, the best probability of success of $\ALG$ is then:
\begin{align*}
	\frac{1}{2}+\frac{1}{2} \cdot \tvd{(\QA{x} \mid \dist_0)}{(\QA{x} \mid \dist_1)} &\leq 1-\frac{1}{4} \cdot \exp\paren{-\kl{\QA{x} \mid \dist_0}{\QA{x} \mid \dist_1}} \tag{by the extension of Pinsker's inequality in~\Cref{prop:pinsker}} \\
	&=1-\frac{1}{4} \cdot \exp\paren{-\kl{\Bern{p_0}^q}{\Bern{p_1}^q}} \tag{by the distribution of $\QA{x}$ argued earlier} \\
	&< 1- \frac{1}{4} \cdot \exp\paren{\ln{(4\delta)}} \tag{by~\Cref{clm:kl} as $k \leq m/6$, $\gamma < 1/4$, and thus $p_0,p_1 < 1/3$}
	= 1-\delta. 
\end{align*}
This means that $\ALG$ can  succeed with probability $<1-\delta$ in distinguishing between $\dist_0$ and $\dist_1$. 
Combined with the easy direction of Yao's minimax principle (namely, an averaging principle,~\Cref{prop:yao}), this concludes the proof of~\Cref{lem:ORTP}. 
\end{proof}

Before moving on from this section, we also prove the following auxiliary lemma on the range of $|x|_1$ drawn from $D$, which will be used in our subsequent reductions from $\ORTP$. 

\begin{lemma}\label{lem:distlemma}
	In the distribution $\dist$, 
	\[
		\Pr\paren{\card{x}_1 > (1-\gamma) \cdot k \mid \dist_0} \leq \delta \quad \text{and} \quad \Pr\paren{\card{x}_1 < (1+\gamma) \cdot k \mid \dist_1} \leq \delta.
	\]
\end{lemma}
\begin{proof}
	The proof is by a simple concentration argument using the independence of $x_i$'s for different $i \in [m]$ when $x$ is sampled from either $\dist_0$ or $\dist_1$. 
	
	\paragraph{Distribution $\dist_0$.} In the following, all randomness is with respect to $\dist_0$. We have, 
	\begin{align*}
		\Exp\card{x}_1 = \sum_{i=1}^{m} \Exp\bracket{x_i} = m \cdot p_0 = (1-2\gamma) \cdot k. 
	\end{align*}
	By the independence of the choices of $x$ across indices $i \in [m]$, we can apply the Chernoff bound (\Cref{prop:conc}) and have that 
	\begin{align*}
		\Pr\paren{\card{x}_1 > (1-\gamma) \cdot k} &\leq \Pr\paren{\card{x}_1 - \Exp \card{x}_1 > \frac{\gamma}{1-2\gamma} \cdot \Exp\card{x}_1} \\
		&\leq \exp\paren{-\frac{\gamma^2 \cdot \Exp \card{x}_1}{(1-2\gamma)^2 \cdot 3}} \leq \exp\paren{-{\gamma^2} \cdot \frac{{12}{} \cdot \ln{(1/\delta)}}{(1-2\gamma) \cdot \gamma^2 \cdot 3}} \leq \delta, 
	\end{align*}
	by the bound on $\Exp\card{x}_1$ and the choice of $k$ and $\gamma$ in~\Cref{lem:ORTP}.  This proves the first part. 
	
	\paragraph{Distribution $\dist_1$.} The proof of this part is almost (but not exactly) identical, and we provide it for completeness. In the following, all randomness is with respect to $\dist_1$. We have, 
	\begin{align*}
		\Exp\card{x}_1 = \sum_{i=1}^{m} \Exp\bracket{x_i} = m \cdot p_1 = (1+2\gamma) \cdot k. 
	\end{align*}
	By the independence of the choices of $x$ across indices $i \in [m]$, we can apply Chernoff bound (\Cref{prop:conc}) and have that
	\begin{align*}
		\Pr\paren{\card{x}_1 < (1+\gamma) \cdot k} &\leq \Pr\paren{\Exp \card{x}_1 - \card{x}_1 > \frac{\gamma}{1+2\gamma} \cdot \Exp\card{x}_1}  \\
		&\leq \exp\paren{-\frac{\gamma^2 \cdot \Exp \card{x}_1}{(1+2\gamma)^2 \cdot 3}} \leq \exp\paren{-{\gamma^2} \cdot \frac{{12}{} \cdot \ln{(1/\delta)}}{(1+2\gamma) \cdot \gamma^2 \cdot 3}} \leq \delta, 
	\end{align*}
	by the bound on $\Exp\card{x}_1$ and the choice of $k$ and $\gamma$ in~\Cref{lem:ORTP}.  This proves the second part. 
\end{proof}

\Cref{lem:distlemma} implies that any algorithm that can differentiate whether $\card{x}_1 \geq (1+\gamma) \cdot k$ or $\card{x}_1 \leq (1-\gamma) \cdot k$ with probability $1-\delta$ can also solve $\ORTP$ with probability $1-2\delta$. This is simply because when
$x \sim \dist_{\theta}$ for $\theta \in \set{0,1}$, with probability at most $\delta$, $\card{x}_1$ is not within the ``right'' range for such an algorithm to detect, and with another probability $\delta$, the algorithm may fail to output the correct answer. 
A union bound then implies the bound of $1-2\delta$ on the probability of correctly solving $\ORTP$. We next use this to prove~\Cref{thm:lower}. 

\subsection{Reducing PTP to the H-Index Problem}

We now prove~\Cref{thm:lower} via a reduction from $\ORTP$ and the lower bound we proved for $\ORTP$ in~\Cref{lem:ORTP}. 

Suppose towards a contradiction that there is an algorithm $\ALG_{h}$ for $h$-index 
that with probability $1-\delta/2$ uses $o(n\ln{(1/\delta)}/(\eps^2\hindex{A}))$ queries on input array $A$ and estimates $\hindex{A}$ to within a $(1\pm \eps)$-factor. Given an instance of $\ORTP_{m,k,\gamma}$, we use $\ALG_h$ to solve $\ORTP$ with probability $1-\delta$ in the following algorithm. 

\begin{tbox}
 \texttt{\hyperref[sec:reductionAlg]{PTP-estimator}($x$, $k$, $\gamma$, $\delta$)}. 

\begin{enumerate}
\item Run $\ALG_h$ with parameters $n=m$, $\eps = \gamma$ and error $\delta/2$ on an array $A$ defined as follows: for any query of $\ALG_h$ to $A[i]$ for $i \in [n]$, return $A[i] = (1+\eps) \cdot k$ if $x_i = 1$ and return $0$ otherwise. 

\item If at any point, the number of queries of $\ALG_h$ reaches
			\[
				\tau(n,k,\eps,\delta) = \frac{n \cdot \ln (1/(4\delta))}{24 \varepsilon^2 \cdot k},
			\]
			stop $\ALG_h$ and return \emph{No} as the answer. 

\item If we never stopped $\ALG_h$, return \emph{Yes} if $\ALG_h$ returns $\tilde{h} \geq k-\eps^2 \cdot k$; otherwise return \emph{No}. 
\end{enumerate}
\end{tbox}

It is clear that the worst-case query complexity of \texttt{\hyperref[sec:reductionAlg]{PTP-estimator}} is $< \tau(n,k,\eps,\delta)$ by the condition on the second line of the algorithm. In terms of parameters for 
$\ORTP_{m,k,\gamma}$, this translates to the bound of
	 $\frac{m \cdot \ln{(1/(4\delta))}}{24\,\gamma^2 \cdot k}$
on the {worst-case} query complexity of \texttt{\hyperref[sec:reductionAlg]{PTP-estimator}}. In the following, we will prove that \emph{if} $\ALG_h$ truly exists, then \texttt{\hyperref[sec:reductionAlg]{PTP-estimator}} solves $\ORTP_{m,k,\gamma}$ with probability of success at least $1-\delta$. But, then \texttt{\hyperref[sec:reductionAlg]{PTP-estimator}} contradicts the lower bound of~\Cref{lem:ORTP} -- this implies that $\ALG_h$ cannot exist, and we get our desired lower bound in~\Cref{thm:lower}. 

\begin{lemma}\label{lem:or-estimator-wins}
	\textnormal{\texttt{\hyperref[sec:reductionAlg]{PTP-estimator}}} outputs the correct answer to any instance of $\ORTP_{m,k,\gamma}$ with probability at least $1-\delta$. 
\end{lemma}
\begin{proof} 
	 \Cref{lem:distlemma} implies that any algorithm that can differentiate whether $\card{x}_1 \geq (1+\gamma) \cdot k$ or $\card{x}_1 \leq (1-\gamma) \cdot k$ with probability $1-\delta/2$ can also solve $\ORTP$ with probability $1-\delta$. Therefore, it is sufficient to prove that \texttt{\hyperref[sec:reductionAlg]{PTP-estimator}} outputs \emph{Yes} when $\card{x}_1 \geq (1+\gamma) \cdot k$ and \emph{No} when $\card{x}_1 \leq (1-\gamma) \cdot k$ with probability at least $1-\delta/2$. We consider each case of the right answer to $\ORTP$ separately.
	
	\paragraph*{Case I:} Suppose first that the input $x$ to $\ORTP$ is a \emph{Yes}-instance, meaning that $\card{x}_1 \geq (1+\gamma) \cdot k$. Consider the array $A$ \emph{implicitly} constructed by \texttt{\hyperref[sec:reductionAlg]{PTP-estimator}}. Given that $\eps = \gamma$, 
	$A$ contains at least $(1+\eps)\cdot k$ entries each with a value of at least $(1+\eps) \cdot k$. Moreover, it does not contain any entry with a value larger than $(1+\eps)\cdot k$. Thus, we have $\hindex{A} = (1+\eps) \cdot k$. By the guarantee of 
	$\ALG_h$ on its correctness and since $\hindex{A} > k$, the probability that $\ALG_h$ outputs a value 
	\[
	\tilde{h} < \hindex{A} - \eps \cdot \hindex{A} = (1+\eps) \cdot k - \eps \cdot k - \eps^2 \cdot k = k - \eps^2 \cdot k
	\]
	or makes more than $\tau(n,k,\eps,\delta)$ queries on $A$ and thus we stop it is at most $\delta/2$. 
	
	\paragraph*{Case II:} Suppose now that the input $x$ to $\ORTP$ is a \emph{No}-instance, meaning that $\card{x}_1 \leq (1-\gamma) \cdot k$. 
	Consider the array $A$ \emph{implicitly} constructed by \texttt{\hyperref[sec:reductionAlg]{PTP-estimator}}. Given that $\eps = \gamma$, 
	$A$ contains at most $(1-\eps)\cdot k$ non-zero entries, so $\hindex{A} \leq (1-\eps) \cdot k$. Thus, by the guarantee of 
	$\ALG_h$ on its correctness, the probability that $\ALG_h$ outputs a value 
	\[
	\tilde{h} \geq k - \eps^2 \cdot k = (1-\eps) \cdot k + \eps \cdot k - \eps^2 \cdot k \geq \hindex{A} + \eps \cdot \hindex{A}
	\]
	is at most $\delta/2$. This means that \emph{if} we do not stop $\ALG_h$ (because it has made too many queries), the output will only be wrong with probability at most $\delta/2$. But now note that we do not have any particular guarantee on the probability
	that we stop $\ALG_h$ as it is possible that $\hindex{A}$ is much less than $k$ and thus the bound of $o(n\ln{(1/\delta)}/(\eps^2\hindex{A}))$ on the queries of $\ALG_h$ will still be way less than $\tau(n,k,\eps,\delta)$. Nevertheless, even if we stop the algorithm, we output \emph{No} as the answer and thus make no error here. Thus, in this case also, the probability of outputting a wrong answer is $\delta/2$ at most as desired. 
	
	This concludes the proof of~\Cref{lem:or-estimator-wins}. 
\end{proof}

\Cref{thm:lower} now follows immediately from~\Cref{lem:ORTP} and~\Cref{lem:or-estimator-wins} as argued earlier.

\section{Triangle Counting Problem}\label{sec:lowerTriangle}  

In this section, we switch from the main theme of our paper which was on the $h$-index problem and instead show an application of our lower bound techniques to the well-studied 
problem of triangle counting using local queries. 

\begin{problem}
In $\TCP_{n, m,\varepsilon}$, for integers $n, m \geq 1$ and parameter $\varepsilon \in (0,1)$, we are given an undirected graph $G=(V,E)$ with $n$ vertices and $m$ edges, and the goal is to estimate 
the number of triangles, namely, cliques on three vertices, in $G$ to within a $(1\pm \eps)$-factor. In order to do this, we can make the following queries to the graph: 
\begin{enumerate}
	\item \textit{Degree queries:} Given a vertex $v \in V$, return the degree of $v$ $(\deg(v))$. 
	\item \textit{Neighbor queries:} Given a vertex $v \in V$ and $i \in [n]$, return the $i^{th}$ neighbor of $v$ if $i \leq \deg(v)$ and ``None" otherwise. 
	\item \textit{Pair queries:} Given two vertices $u, v \in V$, return $1$ if $(u,v) \in E$ and $0$ otherwise. 
	\item \textit{Edge-sample queries:} Return an edge $e \in E$ independently and uniformly at random. 
\end{enumerate} 
\end{problem}

We refer the reader to~\cite{EdenLRS15,EdenRS18,EdenR18b,AssadiKK19} and references therein for more on the background of this problem. Here, we only note that~\cite{EdenLRS15} designed an algorithm for this problem
with time complexity 
	$O^*(\frac{n}{t^{1/3}} + \frac{m^{3/2}}{t})$,
where $t$ is the number of triangles and $O^*$ hides the dependence on $\eps$, error probability $\delta$, and logarithmic factors in $n$. The algorithm of~\cite{EdenLRS15} only requires the first three types of queries mentioned above, 
which is generally considered the baseline for sublinear time algorithms and is referred to as the general query model. Later, by using the fourth type of query also,~\cite{AssadiKK19} obtained an algorithm for this problem
with time complexity
	$O(\frac{m^{3/2} \cdot \ln{(1/\delta)}}{\eps^2 \cdot t})$ 
(the algorithm of~\cite{AssadiKK19} extends to counting \emph{all} subgraphs, not just triangles, with a runtime depending on the fractional edge cover of the subgraph we are counting).

On the lower bound front,~\cite{EdenR18b}, building on~\cite{EdenLRS15}, proved a lower bound of 
$\Omega(\frac{m^{3/2}}{t})$
for the triangle counting problem under the four queries mentioned. This lower bound, however, only holds for some constant $\eps$ and $\delta$ and does not incorporate the dependence on them. 

In this section, using our lower bound for the $\ORTP$ problem in~\Cref{lem:ORTP}, we will improve the lower bound of~\cite{EdenR18b} and obtain a lower bound that matches the algorithmic bounds of~\cite{AssadiKK19}.
\begin{theorem} \label{thm:Tri}
Any algorithm that, given access to an undirected graph $G=(V,E)$ through degree, neighbor, pair, and edge-sample queries, approximation parameter $\varepsilon \in (0,1/4)$, and confidence parameter $\delta \in (0,1/100)$, outputs an estimate $\tilde{t}$ of the number of triangles, $t$, in $G$ such that 
$\Pr(|\tilde{t} - t | \leq \varepsilon \cdot t) \geq 1-\delta$
requires
\[
\Omega(\min(m, \dfrac{m^{3/2} \cdot \ln (1/\delta)}{\varepsilon^2 \cdot t}))
\]
queries to the graph provided that $t = o(\eps \cdot m)$. 
\end{theorem}
\Cref{thm:Tri} now settles the asymptotic complexity of the triangle counting problem in all parameters involved. 

\begin{remark}
	For concreteness, we focused on proving a lower bound only for the triangle counting problem as a representative of the wider family of subgraph counting problems. However, 
	by using our $\ORTP$ technique in place of the lower bound arguments in~\cite{EdenRS18} and~\cite{AssadiKK19},  one can also extend their lower bounds to asymptotically optimal bounds (matching the algorithm of~\cite{AssadiKK19}) 
	for larger cliques as well as odd-cycles.
\end{remark}

Similarly to the $h$-index problem, we prove~\Cref{thm:Tri} via a reduction from $\ORTP$ and our lower bound for that problem in~\Cref{lem:ORTP}.

\begin{remark}
To avoid confusion, in the rest of this proof, we use $m$ to denote the number of edges in the triangle counting problem and instead use $M$ (in place of the original $m$) 
for the dimension of the $\ORTP$ problem. 
\end{remark}

Suppose towards a contradiction that there is an algorithm $\ALG_{t}$ for triangle counting that queries input undirected graph, $G$, and estimates $t$ to within a $(1 \pm \eps)$-factor with probability at least $1-\delta/2 $ using $o(m^{3/2}\ln(1/\delta)/(\varepsilon^2t))$ queries. Given an instance of $\ORTP_{M,k,\gamma}$, we use $\ALG_t$ to solve $\ORTP$ with probability $1-\delta$.
Define $M = (\sqrt{m}/2)^2 = m/4$. We define a mapping from inputs of $\ORTP$, $x \in \{0,1\}^M$, to $G_x(V,E)$ on $n=2\sqrt{m}$ vertices and $m$ edges. 
\begin{itemize}
\item Let the vertices of $G_x$ consist of two sets, $U \cup V$, such that $U = \{u_1, ..., u_{\sqrt{m}}\}$ and $V = \{v_1, ..., v_{\sqrt{m}}\}$. There is no overlap between the two sets, so $U \cap V = \emptyset$. Let $U$ consist of two sets, $U_1 \cup U_2$, such that $U_1 = \{u_1, ..., u_{\sqrt{m}/2}\}$ and $U_2 = \{u_{\sqrt{m}/2 + 1}, ..., u_{\sqrt{m}}\}$. Similarly, let $V$ consist of two sets, $V_1 \cup V_2$, such that $V_1 = \{v_1, ..., v_{\sqrt{m}/2}\}$ and $V_2 = \{v_{\sqrt{m}/2 + 1}, ..., v_{\sqrt{m}}\}$.
\item We view $x$ as being indexed by pairs $i \in [\sqrt{m}/2], j \in [\sqrt{m}/2 + 1, \sqrt{m}]$ such that $i < j$. Now, we add edges in the following way. If $x_{ij} = 1$, $G_x$ contains edges $(u_i,u_j) \in U_1 \times U_2$ and $(v_i,v_j) \in V_1 \times V_2$. If $x_{ij} = 0$, $G_x$ contains edges $(u_i,v_j) \in U_1 \times V_2$ and $(v_i,u_j) \in V_1 \times U_2$. Additionally, for each vertex $u_1 \in U_1$ and $v_1 \in V_1$, $G_x$ contains edge $(u_1,v_1)$. For each vertex $u_2 \in U_2$ and $v_2 \in V_2$, $G_x$ contains edge $(u_2,v_2)$. There are no other edges that are added. 
\end{itemize}

See~\Cref{fig:graphMapping} for an illustration. 
\begin{figure}[h!]
	\centering
	\includegraphics[width=0.3\textwidth]{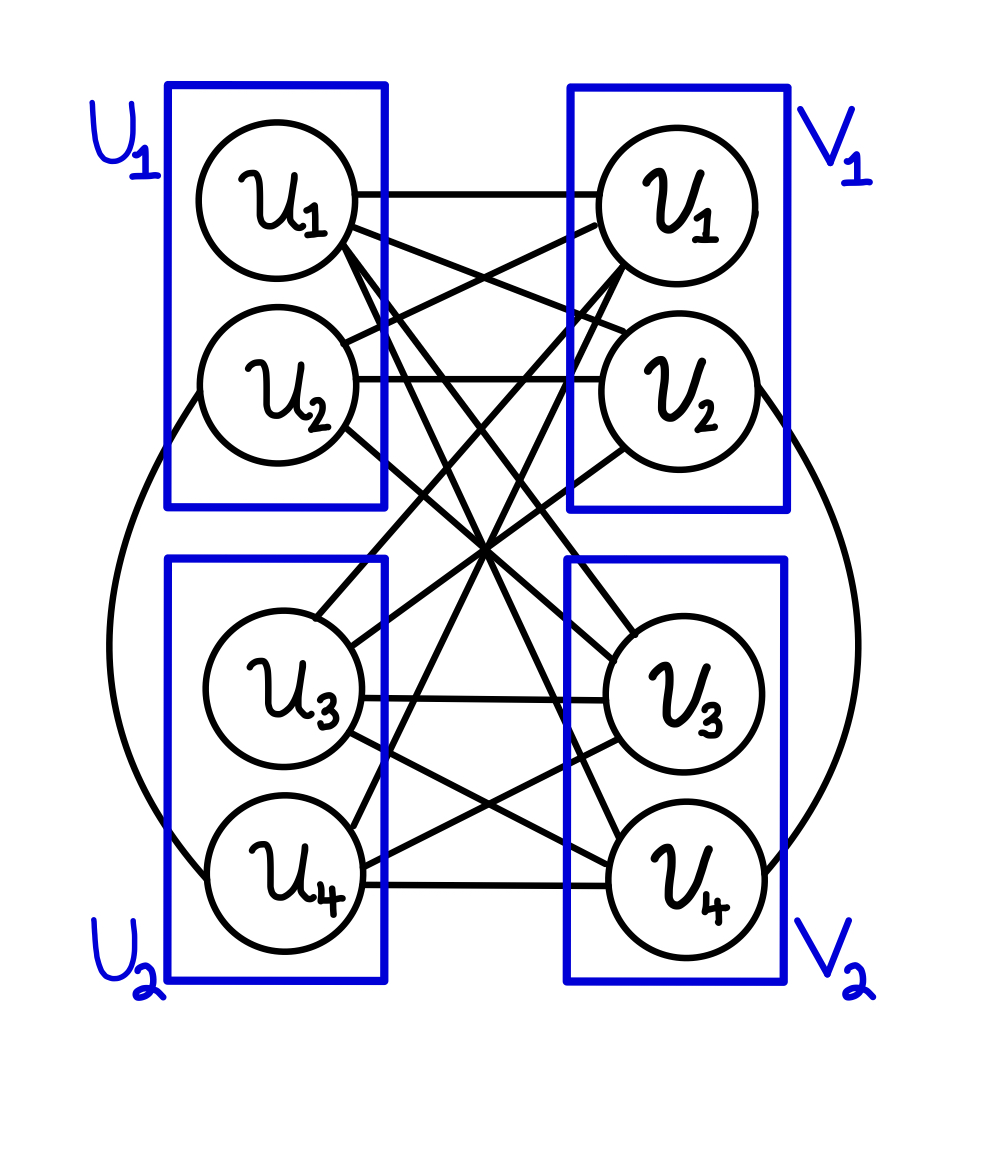}
	\caption{The graph $G_x$ for $x = 0001$. The bits are indexed by the vertex pairs $(13, 14, 23, 24)$. }
	\label{fig:graphMapping}
\end{figure} 

We now have the following reduction.  
\begin{tbox}
\texttt{\hyperref[alg:P2]{PTP-estimator-two}}($x \in \{0,1\}^M$, $k$, $\gamma$, $\delta$)
\begin{enumerate}
\item 
Run $\ALG_{t}$ with parameters $n = 2\sqrt{m}$, $m = 4M$, $\varepsilon = \gamma$, and error $\delta/2$ on an undirected graph $G$ defined as follows: 
\begin{itemize}
\item \textit{Degree queries.} For any degree query of $\ALG_{t}$, return $\sqrt{m}$. 

\item \textit{Neighbor queries.} For any neighbor query of $\ALG_{t}$, do the following. Assume w.l.o.g. that we get a vertex $u_i \in U_1$ and want to find the $k^{th}$ neighbor. If $k \leq \sqrt{m}/2$, return $v_i$. Otherwise, set $j \leftarrow k$. Then, if $x_{ij}$ is $1$, return $u_j$; else, $v_j$. 

\item \textit{Pair queries.} For any pair query of $\ALG_{t}$, if an edge between a vertex $u \in U_1$ and a vertex $v \in V_1$ or between $u\in U_2$ and $v \in V_2$ is queried, return $1$. If an edge between any two vertices in $U_1$, $U_2$, $V_1$, or $V_2$ is queried, return $0$. Else, for some query $(u_i,v_j)$ such that $i < j$, return $\neg x_{ij}$. For some query $(u_i,u_j)$ such that $i<j$, return $x_{ij}$.  

\item \textit{Edge-sample queries.} For any random edge-sample query made by $\ALG_{t}$, uniformly at random pick a vertex $v \in V$ and then uniformly at random pick one of its neighbors $u$. Return the edge $(u, v)$.  
\end{itemize}
\item If at any point, the number of queries of $\ALG_{t}$ reaches 
\[
				\tau(m,k,\eps,\delta) = \frac{m \cdot \ln (1/(4\delta))}{9600 \varepsilon^2 \cdot k},
			\] stop $\ALG_{t}$ and return \emph{No} as the answer. 

\item 
If we never stopped $\ALG_{t}$, return \emph{Yes} if $\ALG_{t}$ returns $\tilde{t} \geq 2k(\sqrt{m}-2)(1-\varepsilon^2)$; otherwise, return \emph{No}. 
\end{enumerate}
\end{tbox}

It is clear that the worst-case query complexity of \texttt{\hyperref[alg:P2]{PTP-estimator-two}} is $<\tau(m,k,\eps,\delta)$. In terms of parameters for $\ORTP_{M,k,\gamma}$, this translates to the bound of 
$\dfrac{M \cdot \ln(1/(4\delta))}{24 \gamma^2 \cdot k}$
on the worst-case query complexity of \texttt{\hyperref[alg:P2]{PTP-estimator-two}}. In the following, we will prove that \emph{if} $\ALG_t$ exists, then \texttt{\hyperref[alg:P2]{PTP-estimator-two}} solves $\ORTP_{M,k,\gamma}$ with probability of success at least $1-\delta$. But then, \texttt{\hyperref[alg:P2]{PTP-estimator-two}} contradicts the lower bound of~\Cref{lem:ORTP} which implies that $\ALG_t$ cannot exist, and we get our desired lower bound in~\Cref{thm:Tri}. 

We note that in the following lemma, the lower bound on $k$ and lower bound on $\eps$ is benign as otherwise the $\Omega(m)$ part of our lower bound in~\Cref{thm:Tri} should instead kick in. 

\begin{lemma} \label{lem:ortpEst2Proof}
\textnormal{\texttt{\hyperref[alg:P2]{PTP-estimator-two}}} outputs the correct answer to any instance of $\ORTP_{M, k, \gamma}$ with probability at least $1-\delta$ as long as $k = \omega(\ln{(1/\delta)}/\eps^2)$,  $k=o(\eps \cdot m)$, and $\eps = \omega(1/\sqrt{m})$.  
\end{lemma}
\begin{proof}
\Cref{lem:distlemma} implies that any algorithm that can differentiate whether $\card{x}_1 \geq (1+\gamma) \cdot k$ or $\card{x}_1 \leq (1-\gamma) \cdot k$ with probability $1-\delta/2$ can also solve $\ORTP$ with probability $1-\delta$. Therefore, it is sufficient to prove that \texttt{\hyperref[alg:P2]{PTP-estimator-two}} outputs \emph{Yes} when $\card{x}_1 \geq (1+\gamma) \cdot k$ and \emph{No} when $\card{x}_1 \leq (1-\gamma) \cdot k$ with probability at least $1-\delta/2$. 

Within $G_x$, we will define \textbf{red edges}. Let the red edges include any edges between any two vertices $\in U_1$. The set of red edges will also include any edges between any two vertices $\in V_1$. For every vertex $v$, we define \textbf{reddeg($v$)} as the number of red edges incident on $v$. 

We consider each case of the right answer to $\ORTP$ separately. 
 
\paragraph*{Case I:} Suppose first that the input $x$ to $\ORTP$ is a \emph{Yes}-instance, meaning that for each index $i \in [M]$, $x_i$ was set to $1$ independently with probability $(1+2\gamma) \cdot k/M$. Consider the graph $G$ \emph{implicitly} constructed by \texttt{\hyperref[alg:P2]{PTP-estimator-two}}. For every bit set to $1$ in $x$, there are two red edges in $G_x$. Each red edge $(u,v)$ creates $(\sqrt{m}-2) - \text{reddeg}(u) - \text{reddeg}(v)$ triangles.

We want to ensure that in the \emph{Yes}-instance, there are enough triangles. We first lower bound the total number of red edges. Since the number of red edges corresponds to $|x|_1$, we can use \Cref{lem:distlemma}. By the choice of $k = \omega(\ln{(1/\delta)}/\eps^2)$, we can see that the probability that $|x|_1 < (1+\gamma) \cdot k $ is bounded by $\delta/2$. Now, we bound for each edge, $(u, v)$, $ \text{reddeg}(u) + \text{reddeg}(v)$. Let us first bound the number of red edges incident on each vertex.

\begin{claim} \label{claim:redV}
When $x$ is a \emph{Yes}-instance, for each vertex $v$, $\Pr(\emph{reddeg}(v) > \eps/3 \cdot \sqrt{m}) \leq \delta/\sqrt{m}$. 
\end{claim}
\begin{proof}
For each vertex $v$, the probability of an edge incident on it being red is $(1+2\vareps) \cdot k/(m/4)$ and there are potentially $\sqrt{m}/2$ red edges. Therefore, $\Exp[\text{reddeg}(v)] =  (1+2\vareps) \cdot k/(m/4) \cdot \sqrt{m}/2$. By the lower bound on $k$, $\Exp[\text{reddeg}(v)] \leq \vareps/4 \cdot \sqrt{m}$. We now use the Chernoff bound (\Cref{prop:conc}) to bound the probability that reddeg($v$) is too large and have 
\begin{align*}
\Pr(\text{reddeg}(v) > \eps/3 \cdot \sqrt{m}) \leq \exp(-\dfrac{(1/3)^2 \cdot \Exp[\text{reddeg}(v)]}{3}) \leq \delta/\sqrt{m}
\end{align*}
where the last inequality is because of the lower bound on $\vareps$. 
\end{proof}

~\cref{claim:redV} implies that any edge $(u,v)$, $ \text{reddeg}(u) + \text{reddeg}(v)$ is at most $2/3 \cdot \vareps/\sqrt{m}$. Thus, by the guarantee of $\ALG_t$ on its correctness, the probability that $\ALG_t$ outputs a value 
\begin{align*}
	\tilde{t} < t - \varepsilon \cdot t \leq 2(\sqrt{m}-2) (1+\varepsilon) \cdot k - \varepsilon \cdot 2(\sqrt{m}-2) (1+\varepsilon) \cdot k = 2k(\sqrt{m}-2)(1-\varepsilon^2)
\end{align*}      
is at most $\delta/2$. This means that \emph{if} we do not stop $\ALG_t$ (because it has made too many queries), the output will only be wrong with probability at most $\delta/2$. Additionally, since $t/(2(\sqrt{m}-2)) > k$ and the number of queries made by $\ALG_t$ is supposed to be $o(m^{3/2}\ln(1/\delta)/(\varepsilon^2t))$, $\ALG_t$ will never make more than $\tau(m,k,\eps,\delta)$ queries on $G$. Therefore, in this case, the probability of outputting a wrong answer is at most $\delta/2$ as desired.  

\paragraph*{Case II:} Suppose instead that the input $x$ to $\ORTP$ is a \emph{No}-instance, meaning that for each index $i \in [M]$, $x_i$ was set to $1$ independently with probability $(1-2\gamma) \cdot k/M$. Consider the graph $G$ \emph{implicitly} constructed by \texttt{\hyperref[alg:P2]{PTP-estimator-two}}. Every red edge can create at most $(\sqrt{m}-2)$ triangles with vertices on the other side of the bipartition. 

We first bound the total number of red edges. Since the number of red edges corresponds to $|x|_1$, we can use \Cref{lem:distlemma}. By the choice of $k = \omega(\ln{(1/\delta)}/\eps^2)$, we can see that the probability that $|x|_1 > (1-\gamma) \cdot k $ is bounded by $\delta/2$. Therefore, by the guarantee of $\ALG_t$ on its correctness, the probability that $\ALG_t$ outputs a value 
\begin{align*}
	\tilde{t} \geq 2k(\sqrt{m}-2) (1-\varepsilon^2) = 2(\sqrt{m} - 2) (1-\varepsilon) \cdot k + \varepsilon \cdot 2(\sqrt{m} -2) (1-\varepsilon) \cdot k  \geq t + \varepsilon \cdot t 
\end{align*}      
is at most $\delta/2$. This means that \emph{if} we do not stop $\ALG_t$ (because it has made too many queries), the output will only be wrong with probability at most $\delta/2$. But now note that we do not have any particular guarantee on the probability that we stop $\ALG_t$ since it is possible that $t/(2(\sqrt{m}-2))$ is much less than $k$ and thus the bound of $o(m^{3/2}\ln(1/\delta)/(\varepsilon^2t))$ on the queries of $\ALG_t$ will still be much less than $\tau(m,k,\eps,\delta)$. Nevertheless, even if we stop the algorithm, we output \emph{No} as the answer and thus make no error here. Thus, in this case also, the probability of outputting a wrong answer is $\delta/2$ at most as desired. 

This concludes the proof of~\Cref{lem:ortpEst2Proof}. 

\end{proof}
\pagebreak

\subsection*{Acknowledgements} 

We thank Janani Sundaresan for helpful feedback on the presentation of our paper. We are also grateful to the anonymous reviewers of APPROX 2022 for their helpful feedback on previous work and the presentation of this paper.

\bibliographystyle{alpha}
\bibliography{general}

\newcommand{\etalchar}[1]{$^{#1}$}
\begin{thebibliography}{BGMP21}

\bibitem[AKK19]{AssadiKK19}
Sepehr Assadi, Michael Kapralov, and Sanjeev Khanna.
\newblock A simple sublinear-time algorithm for counting arbitrary subgraphs
  via edge sampling.
\newblock In Avrim Blum, editor, {\em 10th Innovations in Theoretical Computer
  Science Conference, {ITCS} 2019, January 10-12, 2019, San Diego, California,
  {USA}}, volume 124 of {\em LIPIcs}, pages 6:1--6:20. Schloss Dagstuhl -
  Leibniz-Zentrum f{\"{u}}r Informatik, 2019.

\bibitem[AS22]{AssadiS22}
Sepehr Assadi and Vihan Shah.
\newblock An asymptotically optimal algorithm for maximum matching in dynamic
  streams.
\newblock In Mark Braverman, editor, {\em 13th Innovations in Theoretical
  Computer Science Conference, {ITCS} 2022, January 31 - February 3, 2022,
  Berkeley, CA, {USA}}, volume 215 of {\em LIPIcs}, pages 9:1--9:23. Schloss
  Dagstuhl - Leibniz-Zentrum f{\"{u}}r Informatik, 2022.

\bibitem[BdW02]{BuhrmanW02}
Harry Buhrman and Ronald de~Wolf.
\newblock Complexity measures and decision tree complexity: a survey.
\newblock {\em Theor. Comput. Sci.}, 288(1):21--43, 2002.

\bibitem[BFP{\etalchar{+}}73]{medianOfMedians}
Manuel Blum, Robert~W. Floyd, Vaughan~R. Pratt, Ronald~L. Rivest, and
  Robert~Endre Tarjan.
\newblock Time bounds for selection.
\newblock {\em J. Comput. Syst. Sci.}, 7(4):448--461, 1973.

\bibitem[BGK{\etalchar{+}}22]{BhattacharyaGKL22}
Sayan Bhattacharya, Fabrizio Grandoni, Janardhan Kulkarni, Quanquan~C. Liu, and
  Shay Solomon.
\newblock Fully dynamic ({\(\Delta\)} +1)-coloring in \emph{O}(1) update time.
\newblock {\em {ACM} Trans. Algorithms}, 18(2):10:1--10:25, 2022.

\bibitem[BGMP21]{BishnuGMP21}
Arijit Bishnu, Arijit Ghosh, Gopinath Mishra, and Manaswi Paraashar.
\newblock Query complexity of global minimum cut.
\newblock In Mary Wootters and Laura Sanit{\`{a}}, editors, {\em Approximation,
  Randomization, and Combinatorial Optimization. Algorithms and Techniques,
  {APPROX/RANDOM} 2021, August 16-18, 2021, University of Washington, Seattle,
  Washington, {USA} (Virtual Conference)}, volume 207 of {\em LIPIcs}, pages
  6:1--6:15. Schloss Dagstuhl - Leibniz-Zentrum f{\"{u}}r Informatik, 2021.

\bibitem[BKSV14]{BravermanKSV14}
Vladimir Braverman, Jonathan Katzman, Charles Seidell, and Gregory Vorsanger.
\newblock An optimal algorithm for large frequency moments using
  o(n{\^{}}(1-2/k)) bits.
\newblock In Klaus Jansen, Jos{\'{e}} D.~P. Rolim, Nikhil~R. Devanur, and
  Cristopher Moore, editors, {\em Approximation, Randomization, and
  Combinatorial Optimization. Algorithms and Techniques, {APPROX/RANDOM} 2014,
  September 4-6, 2014, Barcelona, Spain}, volume~28 of {\em LIPIcs}, pages
  531--544. Schloss Dagstuhl - Leibniz-Zentrum f{\"{u}}r Informatik, 2014.

\bibitem[CT06]{CoverT06}
Thomas~M. Cover and Joy~A. Thomas.
\newblock {\em Elements of information theory {(2.} ed.)}.
\newblock Wiley, 2006.

\bibitem[DP09]{DubhashiP09}
Devdatt~P. Dubhashi and Alessandro Panconesi.
\newblock {\em Concentration of Measure for the Analysis of Randomized
  Algorithms}.
\newblock Cambridge University Press, 2009.

\bibitem[EJP{\etalchar{+}}18]{edenJain18}
Talya Eden, Shweta Jain, Ali Pinar, Dana Ron, and C.~Seshadhri.
\newblock Provable and practical approximations for the degree distribution
  using sublinear graph samples.
\newblock In Pierre{-}Antoine Champin, Fabien Gandon, Mounia Lalmas, and
  Panagiotis~G. Ipeirotis, editors, {\em Proceedings of the 2018 World Wide Web
  Conference on World Wide Web, {WWW} 2018, Lyon, France, April 23-27, 2018},
  pages 449--458. {ACM}, 2018.

\bibitem[ELRS15]{EdenLRS15}
Talya Eden, Amit Levi, Dana Ron, and C.~Seshadhri.
\newblock Approximately counting triangles in sublinear time.
\newblock In Venkatesan Guruswami, editor, {\em {IEEE} 56th Annual Symposium on
  Foundations of Computer Science, {FOCS} 2015, Berkeley, CA, USA, 17-20
  October, 2015}, pages 614--633. {IEEE} Computer Society, 2015.

\bibitem[EMR21]{EdenMR21}
Talya Eden, Saleet Mossel, and Ronitt Rubinfeld.
\newblock Sampling multiple edges efficiently.
\newblock In Mary Wootters and Laura Sanit{\`{a}}, editors, {\em Approximation,
  Randomization, and Combinatorial Optimization. Algorithms and Techniques,
  {APPROX/RANDOM} 2021, August 16-18, 2021, University of Washington, Seattle,
  Washington, {USA} (Virtual Conference)}, volume 207 of {\em LIPIcs}, pages
  51:1--51:15. Schloss Dagstuhl - Leibniz-Zentrum f{\"{u}}r Informatik, 2021.

\bibitem[ER18a]{EdenR18b}
Talya Eden and Will Rosenbaum.
\newblock Lower bounds for approximating graph parameters via communication
  complexity.
\newblock In Eric Blais, Klaus Jansen, Jos{\'{e}} D.~P. Rolim, and David
  Steurer, editors, {\em Approximation, Randomization, and Combinatorial
  Optimization. Algorithms and Techniques, {APPROX/RANDOM} 2018, August 20-22,
  2018 - Princeton, NJ, {USA}}, volume 116 of {\em LIPIcs}, pages 11:1--11:18.
  Schloss Dagstuhl - Leibniz-Zentrum f{\"{u}}r Informatik, 2018.

\bibitem[ER18b]{EdenR18}
Talya Eden and Will Rosenbaum.
\newblock On sampling edges almost uniformly.
\newblock In Raimund Seidel, editor, {\em 1st Symposium on Simplicity in
  Algorithms, {SOSA} 2018, January 7-10, 2018, New Orleans, LA, {USA}},
  volume~61 of {\em OASIcs}, pages 7:1--7:9. Schloss Dagstuhl - Leibniz-Zentrum
  f{\"{u}}r Informatik, 2018.

\bibitem[ERS18]{EdenRS18}
Talya Eden, Dana Ron, and C.~Seshadhri.
\newblock On approximating the number of k-cliques in sublinear time.
\newblock In Ilias Diakonikolas, David Kempe, and Monika Henzinger, editors,
  {\em Proceedings of the 50th Annual {ACM} {SIGACT} Symposium on Theory of
  Computing, {STOC} 2018, Los Angeles, CA, USA, June 25-29, 2018}, pages
  722--734. {ACM}, 2018.

\bibitem[ERS20]{EdenRS20}
Talya Eden, Dana Ron, and C.~Seshadhri.
\newblock Faster sublinear approximation of the number of \emph{k}-cliques in
  low-arboricity graphs.
\newblock In Shuchi Chawla, editor, {\em Proceedings of the 2020 {ACM-SIAM}
  Symposium on Discrete Algorithms, {SODA} 2020, Salt Lake City, UT, USA,
  January 5-8, 2020}, pages 1467--1478. {SIAM}, 2020.

\bibitem[FGP20]{FichtenbergerG020}
Hendrik Fichtenberger, Mingze Gao, and Pan Peng.
\newblock Sampling arbitrary subgraphs exactly uniformly in sublinear time.
\newblock In Artur Czumaj, Anuj Dawar, and Emanuela Merelli, editors, {\em 47th
  International Colloquium on Automata, Languages, and Programming, {ICALP}
  2020, July 8-11, 2020, Saarbr{\"{u}}cken, Germany (Virtual Conference)},
  volume 168 of {\em LIPIcs}, pages 45:1--45:13. Schloss Dagstuhl -
  Leibniz-Zentrum f{\"{u}}r Informatik, 2020.

\bibitem[GMM17]{hIndexStreaming}
Priya Govindan, Morteza Monemizadeh, and S.~Muthukrishnan.
\newblock Streaming algorithms for measuring h-impact.
\newblock In {\em Proceedings of the 36th ACM SIGMOD-SIGACT-SIGAI Symposium on
  Principles of Database Systems}, PODS '17, page 337–346, New York, NY, USA,
  2017. Association for Computing Machinery.

\bibitem[Gol17]{Goldreich17}
Oded Goldreich.
\newblock {\em Introduction to Property Testing}.
\newblock Cambridge University Press, 2017.

\bibitem[GS02]{GibbsS02}
Alison~L Gibbs and Francis~Edward Su.
\newblock On choosing and bounding probability metrics.
\newblock {\em International statistical review}, 70(3):419--435, 2002.

\bibitem[Hir05]{Hirsch05}
Jorge~E. Hirsch.
\newblock An index to quantify an individual's scientific research output.
\newblock {\em Proc. Natl. Acad. Sci. {USA}}, 102(46):16569--16572, 2005.

\bibitem[HP20]{HenzingerP20}
Monika Henzinger and Pan Peng.
\newblock Constant-time dynamic ({\(\Delta\)}+1)-coloring.
\newblock In Christophe Paul and Markus Bl{\"{a}}ser, editors, {\em 37th
  International Symposium on Theoretical Aspects of Computer Science, {STACS}
  2020, March 10-13, 2020, Montpellier, France}, volume 154 of {\em LIPIcs},
  pages 53:1--53:18. Schloss Dagstuhl - Leibniz-Zentrum f{\"{u}}r Informatik,
  2020.

\bibitem[KNP{\etalchar{+}}17]{KapralovNPWWY17}
Michael Kapralov, Jelani Nelson, Jakub Pachocki, Zhengyu Wang, David~P.
  Woodruff, and Mobin Yahyazadeh.
\newblock Optimal lower bounds for universal relation, and for samplers and
  finding duplicates in streams.
\newblock In Chris Umans, editor, {\em 58th {IEEE} Annual Symposium on
  Foundations of Computer Science, {FOCS} 2017, Berkeley, CA, USA, October
  15-17, 2017}, pages 475--486. {IEEE} Computer Society, 2017.

\bibitem[KNW10]{KaneNW10b}
Daniel~M. Kane, Jelani Nelson, and David~P. Woodruff.
\newblock An optimal algorithm for the distinct elements problem.
\newblock In Jan Paredaens and Dirk~Van Gucht, editors, {\em Proceedings of the
  Twenty-Ninth {ACM} {SIGMOD-SIGACT-SIGART} Symposium on Principles of Database
  Systems, {PODS} 2010, June 6-11, 2010, Indianapolis, Indiana, {USA}}, pages
  41--52. {ACM}, 2010.

\bibitem[LW13]{LiW13}
Yi~Li and David~P. Woodruff.
\newblock A tight lower bound for high frequency moment estimation with small
  error.
\newblock In Prasad Raghavendra, Sofya Raskhodnikova, Klaus Jansen, and
  Jos{\'{e}} D.~P. Rolim, editors, {\em Approximation, Randomization, and
  Combinatorial Optimization. Algorithms and Techniques - 16th International
  Workshop, {APPROX} 2013, and 17th International Workshop, {RANDOM} 2013,
  Berkeley, CA, USA, August 21-23, 2013. Proceedings}, volume 8096 of {\em
  Lecture Notes in Computer Science}, pages 623--638. Springer, 2013.

\bibitem[LZZS16]{luZhou16}
Linyuan Lü, Tao Zhou, Qian-Ming Zhang, and H.~Eugene Stanley.
\newblock {The H-index of a network node and its relation to degree and
  coreness}.
\newblock {\em Nature Communications}, 7(1):1--7, April 2016.

\bibitem[NY19]{NelsonY19}
Jelani Nelson and Huacheng Yu.
\newblock Optimal lower bounds for distributed and streaming spanning forest
  computation.
\newblock In {\em Proceedings of the Thirtieth Annual {ACM-SIAM} Symposium on
  Discrete Algorithms, {SODA} 2019, San Diego, California, USA, January 6-9,
  2019}, pages 1844--1860, 2019.

\bibitem[PW11]{PriceW11}
Eric Price and David~P. Woodruff.
\newblock {(1} + eps)-approximate sparse recovery.
\newblock In Rafail Ostrovsky, editor, {\em {IEEE} 52nd Annual Symposium on
  Foundations of Computer Science, {FOCS} 2011, Palm Springs, CA, USA, October
  22-25, 2011}, pages 295--304. {IEEE} Computer Society, 2011.

\bibitem[PW13]{PriceW13}
Eric Price and David~P. Woodruff.
\newblock Lower bounds for adaptive sparse recovery.
\newblock In Sanjeev Khanna, editor, {\em Proceedings of the Twenty-Fourth
  Annual {ACM-SIAM} Symposium on Discrete Algorithms, {SODA} 2013, New Orleans,
  Louisiana, USA, January 6-8, 2013}, pages 652--663. {SIAM}, 2013.

\bibitem[RC16]{RiquelmeC16}
Fabi{\'{a}}n Riquelme and Pablo~Gonzalez Cantergiani.
\newblock Measuring user influence on twitter: {A} survey.
\newblock {\em Inf. Process. Manag.}, 52(5):949--975, 2016.

\bibitem[Sol16]{Solomon16}
Shay Solomon.
\newblock Fully dynamic maximal matching in constant update time.
\newblock In Irit Dinur, editor, {\em {IEEE} 57th Annual Symposium on
  Foundations of Computer Science, {FOCS} 2016, 9-11 October 2016, Hyatt
  Regency, New Brunswick, New Jersey, {USA}}, pages 325--334. {IEEE} Computer
  Society, 2016.

\bibitem[SSP18]{sariyuceSP18}
Ahmet~Erdem Sariy{\"{u}}ce, C.~Seshadhri, and Ali Pinar.
\newblock Local algorithms for hierarchical dense subgraph discovery.
\newblock {\em Proc. {VLDB} Endow.}, 12(1):43--56, 2018.

\bibitem[Tsy09]{paraEstimation}
Alexandre~B. Tsybakov.
\newblock {\em Introduction to Nonparametric Estimation}.
\newblock Springer series in statistics. Springer, 2009.

\bibitem[TT21]{TetekT22}
Jakub T{\v{e}}tek and Mikkel Thorup.
\newblock Sampling and counting edges via vertex accesses.
\newblock {\em arXiv preprint arXiv:2107.03821. To appear in STOC 2022}, 2021.

\bibitem[Yao77]{Yao77}
Andrew~Chi{-}Chih Yao.
\newblock Probabilistic computations: Toward a unified measure of complexity
  (extended abstract).
\newblock In {\em 18th Annual Symposium on Foundations of Computer Science,
  Providence, Rhode Island, USA, 31 October - 1 November 1977}, pages 222--227.
  {IEEE} Computer Society, 1977.

\end{thebibliography}

\end{document}